\documentclass[aps,prl,twocolumn,superscriptaddress,longbibliography]{revtex4-2}

\usepackage[utf8]{inputenc}
\usepackage[T1]{fontenc}
\usepackage{amsmath}
\usepackage{amsthm}
\usepackage[french, german, english]{babel}
\usepackage{physics, quantikz}
\usepackage{amssymb}
\usepackage{mathtools}
\usepackage{bm}
\usepackage[version=4]{mhchem}
\usepackage{graphicx}
\usepackage{xcolor}
\usepackage{subcaption}
\usepackage{enumitem}
\usepackage[section]{placeins}
\usepackage[left=2.5cm,right=2.5cm,top=2cm, bottom=2cm]{geometry}
\usepackage{tikz}
\usetikzlibrary{babel,shapes.geometric, arrows}

\usepackage{hyperref}
\usepackage[capitalise]{cleveref}
\hypersetup{
    colorlinks,
    linkcolor={red!70!black},
    citecolor={blue},
    urlcolor={blue!40!black}
}

\usepackage{bbm}	
\usepackage{bm}     

\renewcommand{\selectlanguage}[1]{}

\DeclareCaptionFont{eightpt}{\fontsize{8pt}{9pt}\selectfont}

\DeclareMathOperator{\polylog}{polylog}

\newcommand{\bigO}{\mathcal O}
\newcommand{\smallo}{o}
\newcommand{\RR}{\mathbbm R}
\newcommand{\NN}{\mathbbm N}

\newcommand{\eps}{\varepsilon}

\newcommand{\cre}[1]{a^{\dag}_{#1}}	
\newcommand{\des}[1]{a_{#1}}			

\newcommand{\T}[2]{T_{#1#2}}

\newcommand{\decayalgebra}{\mathcal{D}}

\newcommand{\nsys}{N}					

\newcommand{\nnuc}{n}

\newcommand{\nanc}{n_a}				
\newcommand{\lowprecbits}{D}

\newcommand{\orb}[1]{\phi_{#1}}			
\renewcommand{\ao}[1]{\chi_{#1}}

\newcommand{\rupper}{r_{\max}}
\newcommand{\localsize}{L}				

\newcommand{\degexp}{K}				

\newcommand{\qft}{\textsc{QFT}}
\tikzset{
  mux/.style={
    draw=none, fill=none, inner sep=0pt,
    path picture={
      \fill[black]
        ([shift={(-1mm,-1mm)}]path picture bounding box.center)
        rectangle
        ([shift={( 1mm, 1mm)}]path picture bounding box.center);
    }
  }
}
\newcommand{\mctrl}[1]{\gate[style=mux]{}\vqw{#1}}

\newtheorem{definition}{Definition}

\newtheorem{lemma}[definition]{Lemma}
\newtheorem{theorem}[definition]{Theorem}

\begin{document}

\title{Sublinear-Depth Quantum Simulation of Electrons with Atomic Orbitals}
\author{Jakob Günther}
\affiliation{\footnotesize Center for Theoretical Physics - a Leinweber Institute, MIT, Cambridge Massachusetts 02139, USA}
\affiliation{\footnotesize Department of Mathematical Sciences, University of Copenhagen, 2100 Copenhagen, Denmark}

\author{Aram W.~Harrow}
\affiliation{\footnotesize Center for Theoretical Physics - a Leinweber Institute, MIT, Cambridge Massachusetts 02139, USA}

\begin{abstract}
Discretizing electronic degrees of freedom into atomic orbitals is the default choice to accurately model electrons in molecules in a compact and flexible way.
We consider the task of simulating such electronic structure models on a quantum computer for $\nsys$ orbitals, keeping the number of orbitals per atom constant.
Atomic orbitals give rise to Hamiltonians with potentially up to $\bigO(\nsys^4)$ terms and little structure, which made it challenging to match the complexity of methods based on plane waves and real-space grids. 
In this work, we show that a Trotter step for atomic orbitals can be implemented in depth $\bigO(\polylog(\nsys))$ without ancillas by combining a rigorous treatment of orbital localization, a hierarchical Hamiltonian decomposition based on the Fast Multipole Method, and shallow quantum Fourier arithmetic circuits.
By bounding the Trotter error, we obtain sublinear simulation depth, and our total gate count $\nsys^{5/3 + o(1)}$ matches the best scaling known for any basis.
\end{abstract}

\maketitle

Simulating electrons in materials and molecules is among the most promising applications of digital quantum computing.
The performance of a quantum simulation algorithm is generally described with respect to three parameters: the evolution time $t$, the accuracy $\eps$, and a system-size parameter $\nsys$.
Most research on electronic structure simulation rests on one of two simulation paradigms: Trotter product formulas and the qubitization framework.
The performance in terms of $t$ and $\eps$ is understood well for both methods. 
For qubitization, the asymptotic scaling is $\bigO(t + \log\small(\eps^{-1}\small))$ \cite{low_hamiltonian_2019}, while the $p$th-order Trotter formula scales as $\bigO(t^{1 + 1/p}\eps^{-1/p})$ \cite{berry_efficient_2007}.
The scaling with $\nsys$, however, depends on the type of interactions and the decisions made in modeling the system.
Tight bounds are often not known, in particular for systems with less structured interactions.
Electronic structure simulation based on atomic orbitals is a primary example of this.
The classical computational framework of molecular electronic structure rests on atomic orbitals, which underpin many central concepts of chemistry, reflecting the insight that atoms remain meaningful subunits within a molecule.
Improving and understanding the system-size dependence is a long-standing problem in the field \cite{aspuru-guzik_simulated_2005,whitfield_simulation_2011,wecker_gate-count_2014,hastings_improving_2015,poulin_trotter_2015,mcclean_exploiting_2014,reiher_elucidating_2017,motta_low_2021,babbush_exponentially_2016,babbush_exponentially_2017,berry_qubitization_2019,von_burg_quantum_2021,lee_even_2021,rocca_reducing_2024,liu_block_2025,gunther_phase_2026,low_denser_2026}.

\begin{table}[t]
	\captionsetup{font=eightpt}
	\fontsize{7.8pt}{7pt}\selectfont
	\begin{tabular}{|c|c|c|c|c|c|} \hline
		ref. & method & dim & ancillas & total gates & depth \\ \hline
		Wecker \cite{wecker_gate-count_2014} & Trotter & 3D & 0 & $\bigO(\nsys^{10})$ & $\bigO(\nsys^{10})$ \\
		Wecker \cite{wecker_gate-count_2014} & Trotter & 3D & $\bigO(\nsys)$ & $\bigO(\nsys^{10})$ & $\bigO(\nsys^9)$ \\
		Hastings \cite{hastings_improving_2015} & Trotter & 3D & 0 & $\bigO(\nsys^9)$ & $\bigO(\nsys^8)$ \\
		Hastings \cite{hastings_improving_2015} & Trotter & 3D & $\bigO(\nsys)$ & $\bigO(\nsys^9)$ & $\bigO(\nsys^7)$ \\
		Poulin \cite{poulin_trotter_2015} & Trotter & 3D & 0 & $\sim\nsys^{6}$ & $\sim \nsys^5$ \\
		McClean \cite{mcclean_exploiting_2014} & Trotter & 3D & 0 & $\sim\nsys^5$ & $\sim \nsys^5$ \\
		Babbush \cite{babbush_exponentially_2016} & Taylor & 3D & $\tilde\bigO(1)$ & $\tilde\bigO(\nsys^5)$ & $\tilde\bigO(\nsys^5)$   \\
		Berry \cite{berry_qubitization_2019} & qubitz. & 1D & $\sim\log(\nsys)$ & $\sim\nsys^{2.3}$ & $\sim\nsys^{2.3}$  \\
		von Burg \cite{von_burg_quantum_2021} & qubitz. & 1D & $\sim\nsys^{1.5}$ & $\sim\nsys^{3.4}$ & $\sim\nsys^{3.4}$ \\
		Lee \cite{lee_even_2021} & qubitz. & 1D & $\sim\nsys$ & $\sim\nsys^{2.1}$ & $\sim\nsys^{2.1}$ \\
		Liu \cite{liu_block_2025} & qubitz. & 3D & $\tilde\bigO(\nsys)$ & $\tilde\bigO(\nsys^3)$ & $\tilde\bigO(\nsys^3)$ \\ 
		\bf{This work} & Trotter & 1D & $0$ & $\nsys^{1 + \smallo(1)}$ & $\nsys^{\smallo(1)}$ \\
		\bf{This work} & Trotter & 2D & $0$ & $\nsys^{3/2 + \smallo(1)}$ & $\nsys^{1/2 + \smallo(1)}$ \\
		\bf{This work} & Trotter & 3D & $0$ & $\nsys^{5/3 + \smallo(1)}$ & $\nsys^{2/3 + \smallo(1)}$ \\\hline
	\end{tabular}
	\caption{\fontsize{7.7pt}{6.5pt}\selectfont $\nsys$-dependence of electronic structure simulation with atomic orbitals for $\Theta(\nsys)$ electrons.
	$O(\cdot), \tilde O(\cdot)$ indicate theoretical upper bounds, $\sim$ indicates empirical results.
	Only clean ancillas are counted.
	Refs.~\cite{berry_qubitization_2019,von_burg_quantum_2021,lee_even_2021} are based on hydrogen chain fits of \cite{lee_even_2021}.
	Lower ancilla scaling is possible for \cite{von_burg_quantum_2021,lee_even_2021} at an increase in total gate complexity \cite{low_denser_2026}.}
		\label{table}
\end{table}

Early works on quantum simulation for molecules employed atomic orbitals and proposed low-order Trotter formulas with direct implementation of all $\bigO(\nsys^4)$ terms, but their cost analysis resulted in impractically large polynomial scaling \cite{wecker_gate-count_2014,mcclean_exploiting_2014,hastings_improving_2015,poulin_trotter_2015}.
Improvements were enabled by a switch to LCU-based simulation algorithms \cite{babbush_exponentially_2016,babbush_exponentially_2017,berry_qubitization_2019} in combination with increasingly sophisticated Hamiltonian decomposition techniques \cite{berry_qubitization_2019,von_burg_quantum_2021,lee_even_2021,rocca_reducing_2024}, and, more recently, by restricting to the low-energy subspace \cite{low_fast_2025} (which is not done in this work).
An overview of select results is given in \Cref{table}.
Parallel lines of research studied the complexity of electronic structure simulation with plane waves \cite{childs_theory_2021,babbushEncodingElectronicSpectra2018,babbushLowDepthQuantumSimulation2018,kivlichanQuantumSimulationElectronic2018,low_interaction_2019,kivlichanImprovedFaultTolerantQuantum2020,su_nearly_2021,mcardle_exploiting_2022,babbush_quantum_2019,rubin_quantum_2024} and in a real-space basis \cite{rubin_quantum_2024,stetina_first-quantized_2025,berry2025quantumsimulationchemistryquantum,low_complexity_2023}.
In contrast to atomic orbitals, these bases are much less tailored to a specific system, necessitating a significant constant overhead in the number of orbitals.
However, their basic forms simplify the Hamiltonian to contain only $\bigO(\nsys^2)$ terms with simple analytic coefficients.
Additionally, a description in first quantization can ameliorate high costs due to basis set sizes \cite{babbush_quantum_2019,rubin_quantum_2024,stetina_first-quantized_2025,berry2025quantumsimulationchemistryquantum}.
In the large-system-size limit, keeping the atom density and orbital resolution bounded by a constant, the method by Berry et al. \cite{berry2025quantumsimulationchemistryquantum} achieves a gate complexity of $\nsys^{5/3+o(1)}$.
Overall, previous results implied that atomic orbitals are not competitive for large-scale simulations.

In this work we show that Trotter-based quantum simulation of electronic structure with atomic orbitals achieves state-of-the-art gate complexity in sublinear depth and without ancillas.
Our first contribution is to formally define a notion of orbital localization that we will prove is satisfied by realistic many-atom systems and that will enable the quantum speedups we discuss below.
We then apply several algorithmic techniques, some new and some adapted, to implement each Trotter step more efficiently, and to improve the naive bounds on Trotter error.
Together, these result in our improved performance described in \Cref{table} and \Cref{thm:main result}.

\textit{Electronic-Structure Model} ---
We consider electrons in the presence of $\nnuc$ nuclei at positions $r_1,\dots,r_{\nnuc}\in\RR^3$ arranged in a $d$-dimensional configuration and with charges $Z_I$.
 Let $\{\orb{p}\}_{p=1}^{\nsys}$ be a set of real-valued orthonormal orbitals to which fermionic creation and annihilation operators $\cre{p}$ and $\des{p}$ are associated (our method extends to complex-valued orbitals, but we keep the exposition simple).
We introduce operators $\T{p}{q} = \cre{p}\des{q} + \cre{q}\des{p}$ for $p\neq q$ and $\T{p}{p} = \cre{p}\des{p}$; then the electronic structure Hamiltonian takes the form
\begin{align}
	H &= h_1 + h_2 = \sum_{p\leq q} h_{pq} \T{p}{q} + \frac12\sum_{\substack{p\leq q\\ r\leq s}} h_{pqrs} \T{p}{q}\T{r}{s} \label{eq:defn hamil}
\end{align}
where the coefficients are defined as
\begin{align*}
	&h_{pq} = \int \orb{p}(r) \left(-\frac{\nabla^2}{2} + v(r)\right)\orb{q}(r)dr - \sum_r \frac{h_{prrq}}{2}\\
	&h_{pqrs} = \int \int \frac{\orb{p}(r_1)\orb{q}(r_1)\orb{r}(r_2) \orb{s}(r_2)}{\abs{r_1 - r_2}}  dr_1 dr_2,
\end{align*}
and $v(r)= -\sum_I\frac{Z_I}{\abs{r-r_I}}$ is the potential of the nuclei.

For molecular systems, orbitals are expressed as linear combinations of atomic orbitals of all nuclei.
Typical choices of atomic orbitals include Slater-type and Gaussian orbitals.   
These possibilities, and others, satisfy the requirements for our definition of atomic orbitals.

\begin{definition}\label[definition]{defn:atomic orbital set}
	Given $\nnuc$ nuclei at positions $(r_1,\dots,r_{\nnuc}) \in \RR^{3\nnuc}$ with minimal separation $\delta>0$, a set of normalized functions $\{\ao{I i}\} \subset L^2(\RR^3)$ for $I\in[\nnuc]$ and $i\in[\localsize]$ is an \emph{atomic-orbital set} with $\localsize$ orbitals per nucleus if $\beta,\zeta,A>0$ exist such that
        \begin{align}
          \abs{\nabla\ao{I i}(r)} \leq \beta \zeta e^{-\zeta\abs{r - r_I}}& \qquad \forall r \in \RR^3 \setminus \{r_I\} \\
          \lambda_{\min}(S)\geq A \qquad &S_{Ii,Jj} = \braket{\ao{Ii}}{\ao{Jj}}
        \end{align}
        where $\lambda_{\min}(S)$ denotes the smallest eigenvalue of $S$.
\end{definition}

The exponential-decay property and normalization imply that $\abs{\ao{Ii}(r)} \leq \beta e^{-\zeta\abs{r-r_I}}$ (see End Matter for details).
For comparison, we note that the decay of the plane-wave dual basis is only $\bigO(\abs{r}^{-1})$.
The second condition quantifies the linear dependence of the atomic-orbital set; $\lambda_{\min}(S)=1$ for orthonormal orbitals and $\lambda_{\min}(S)=0$ for linearly dependent ones.
The latter signifies an inefficient model because it contains redundant degrees of freedom, but this issue is not encountered for atomic orbitals used in practice.
However, whether $\lambda_{\min}(S)$ asymptotically reaches zero, or a value greater than zero, crucially depends on how the limit $\nsys\to\infty$ is approached.
The \textit{continuum limit} keeps $\nnuc$ fixed while increasing $\localsize \to \infty$, and it does not admit a uniform bound $A>0$.
To illustrate this, consider hydrogen orbitals at two neighboring atoms: their decay rate vanishes with increasing principal quantum number, so the corresponding orbitals on the two atoms become nearly identical and $\lambda_{\min}(S)\to 0$ for $\localsize\to\infty$.  
This long-known issue is a cause of numerical instabilities \cite{lowdin_quantum_1956,ahlenius_comments_1973}, and several works studied the decrease of $\lambda_{\min}(S)$ numerically for increasing $\localsize$ \cite{ahlenius_comments_1973,lehtola_curing_2019,hoyvik_spectrum_2020}.
Because modeling the continuum limit with atomic orbitals is inherently problematic we do not further discuss it in this work.

The situation we consider corresponds to the \textit{thermodynamic limit}, where the growth of $\nsys$ and the number of electrons is proportional to the volume of the system, so that the net charge stays close to zero.
In the context of \Cref{defn:atomic orbital set} this is accomplished by $\delta>0$ and by keeping $\localsize$ constant.
The behavior of $\lambda_{\min}(S)$ in this situation has been studied for a cubic lattice with two Gaussian orbitals per site in \cite{aissing_linear_1992}.
An important finding is that after taking the limit $\nnuc\to\infty$ a finite value $A>0$ is obtained, although $A$ decreases rapidly as the product $\delta\zeta$ decreases.
We provide a general $\nnuc$-independent lower bound on $\lambda_{\min}(S)$ that is greater than zero for sufficiently large $\delta\zeta$ (see End Matter for details).
Overall, a lower bound $A>0$ uniform in $\nsys$ in the thermodynamic limit is a justified assumption if $\delta,\beta$ and $\zeta$ are held constant.

When expressed in terms of atomic orbitals, a large part of the Coulomb integrals $h_{pqrs}$ almost vanishes, because the overlap between atomic orbitals decays exponentially with the distance between atoms.
However, atomic orbitals are in general not orthogonal, and significant interatomic spatial overlap is essential to describe chemical bonds. 
In practice, one usually starts from a pre-optimized atomic-orbital set \footnote{See \url{https://www.basissetexchange.org} for an extensive list of widely used basis sets in quantum chemistry.}, and then obtains an orthonormal, but not necessarily local, orbital basis, such as the Hartree--Fock orbitals.
A further localization step reduces the spatial extent of these orbitals, while preserving orthonormality, and many localization methods are formulated as minimization problems.
From a practitioner's viewpoint the existence of sufficiently localized orthonormal orbital bases is not questioned, but existence is \textit{a priori} not guaranteed for the general setup we defined.
While some works on quantum algorithms for chemistry simply start from localized orthonormal orbitals \cite{mcclean_exploiting_2014,babbush_exponentially_2016,liu_block_2025}, we overcome this issue by proving that \emph{Löwdin orthogonalization} \cite{lowdin_nonorthogonality_1950} leads to an orthonormal basis with the desired decay properties.

The change-of-basis matrix from the original basis $\{\ao{Ii}\}$ to the Löwdin basis is $S^{-1/2}$.
The L\"owdin basis is automatically orthonormal, assuming that $S$ is invertible, but if $S$ is well-conditioned then we can say more.  In this case, our \Cref{thm:lowdin} proves that the L\"owdin basis approximately inherits the localization properties of the atomic-orbital basis.
This builds on  a substantial body of work on entry-wise decay bounds for functions of matrices with a decay structure~\cite{jaffard_proprietes_1990,krishtal_localization_2015,benzi_decay_2012,benzi_localization_2016}.
While \cite{benzi_decay_2012} connects this theory to density matrices of many-electron systems and to truncations of Hamiltonians, we are not aware of an application to localization of molecular orbitals in the literature.

\begin{theorem}\label[theorem]{thm:lowdin}
	Given an atomic-orbital set for $\nnuc$ nuclei, the Löwdin orbitals are atomic orbitals with decay parameters $\beta',\zeta'>0$ that are functions of $\zeta,\beta,\localsize,\delta,A$ independent of $\nnuc$. 	
\end{theorem}

Together with nuclear coordinates $\bm{r} \in \RR^{3\nnuc}$ and charges $\bm Z \in [Z_{\max}]^{\nnuc}$, an orthonormal basis $\{\orb{Ii}\}$ determines an electronic-structure Hamiltonian $H = H(\bm{r},\bm{Z},\{\orb{Ii}\})$ via \cref{eq:defn hamil} (with $Z_{\max}$ being a constant, because the periodic table is finite). 
\Cref{thm:lowdin} justifies that $\{\orb{Ii}\}$ is an atomic-orbital set itself if the underlying single-electron Hilbert space was generated from an atomic-orbital set.
Therefore, we are motivated to give the following definition: 

\begin{definition}[Molecular Electronic-Structure Hamiltonians]\label[definition]{defn:mesh}
	$\mathrm{MESH}(\zeta,\beta,\localsize,\delta,A)$ is the family of Hamiltonians $H(\bm{r},\bm{Z},\{\orb{Ii}\})$ with minimal nuclear distance $\delta>0$, and $\{\orb{Ii}\}$ is an orthonormal atomic-orbital set for $\bm{r}$ with local resolution $\localsize \in \NN_{>0}$ and parameters $\zeta,\beta,A>0$.
\end{definition}
This definition applies both to large molecules and to solid-state systems such as lattices.
The central problem of this work is to implement the time evolution $e^{-iHt}$ for $H \in \mathrm{MESH}(\zeta,\beta,\localsize,\delta,A)$ to precision $\eps$ in the limit $\nnuc\to\infty$ and $\eps\to 0$ while keeping other parameters constant.

\textit{Sparsity and low rank} ---
The first step to simplify $H\in \mathrm{MESH}(\zeta,\beta,\localsize,\delta,A)$ is to turn the exponential decay of the atomic orbital basis into sparsity of the Hamiltonian coefficients by means of truncation.
We do this by limiting the support of each atomic orbital to a finite-sized ball of radius $r_{\max}>0$ around its nucleus and evaluate $h_{pq}$ and $h_{pqrs}$ with these modified orbitals.
For truncation error $\eps$ we can set $r_{\max} = \bigO(\log\frac{\nsys}{\eps})$.
It is useful to consider the undirected graph representing the orbital overlaps, where the edge set $E\subset [\nsys]\times[\nsys]$ contains an orbital pair $(p,q)$ only if the corresponding nuclei $P$ and $Q$ satisfy $\abs{r_P - r_Q} < 2r_{\max}$.  
As a minimal distance $\delta$ between nuclei is assumed the graph has maximal degree $\Delta = \bigO\big((\frac{1}{\zeta\delta}\log\frac{\nsys}{\eps})^{d}\big)$. 
Consequently, only $\tilde\bigO(\nsys)$ and $\tilde\bigO(\nsys^2)$ of the coefficients $h_{pq}$ and $h_{pqrs}$ are non-zero, respectively.
The truncation is described in more detail in the End Matter.
In the following discussion it is helpful to view the Coulomb interaction as $\tilde\bigO(\nsys^2)$ interactions between $\tilde\bigO(\nsys)$ orbital pairs, instead of considering the orbitals individually. 
We also remark at this point that instead of truncating, one could take into account the small terms in a randomized fashion as discussed in \cite{hagan_composite_2023,jin_partially_2025,gunther_phase_2026}.

In the next step we take advantage of a low-rank structure of the Coulomb interaction by applying ideas of the Fast Multipole Method (FMM), a well-established classical method to speed up long-range pair interactions \cite{greengard_fast_1987, greengard_rapid_1988, greengard_new_1997,greengard_fast_1994, white_continuous_1994}.
The core idea behind it is that the Coulomb potential is well approximated by low-degree polynomials (multipole expansions) at large distances.
A hierarchical partitioning strategy recursively groups interactions into sets of interactions between spatially separated regions, resulting in a tree-structure with a logarithmic number of levels.
Analytic formulas for translating multipole expansions within and between levels allow for a very efficient construction of the potentials at all levels. 
Moreover, if particles are sufficiently homogeneously distributed, then the number of short-range interactions per particle is constant.
For fixed precision, this allows one to compute the potential at a single point in space in a constant number of steps, instead of scaling linearly as by direct computation.
Previously proposed quantum simulation algorithms based on the FMM follow its algorithmic steps, in particular the translations of multipole expansions are performed quantumly \cite{berry2025quantumsimulationchemistryquantum,wang_polylogarithmic-depth_2025}. 
In contrast, we only take inspiration from the FMM to find a suitable Hamiltonian decomposition.

In our setting, the localized supports of the $\tilde\bigO(N)$ orbital pairs take the role of the positions of point particles in the classical FMM.
The Coulomb interaction is split into short- and long-range contributions, $h_2 = h_2^{\text{short}} + h_2^{\text{long}}$, where $h_2^{\text{short}}$ only contains the $\tilde\bigO(\nsys)$ interactions in which orbital pairs $(p,q)$ and $(r,s)$ are close to each other.
We now describe key features of the hierarchical partitioning, and refer to the original work for details \cite{greengard_fast_1987,greengard_rapid_1988, greengard_new_1997}.
Let the orbitals be enclosed by a cube $[0,2^W]^3 \subset \RR^3$ for integer $W$.
The partitioning of interactions proceeds in levels, indexed by an integer $j$, and the interaction range covered at level $j$ is decreasing as $j$ increases.
At each level the interactions are determined by a partition $B_j$ of the initial cube. 
$B_1$ is the partition into $4\times 4\times 4$ sub-cubes (or \textit{boxes}), and $B_j$ is obtained by dividing each box at level $j-1$ into 8 smaller boxes with side lengths $d_j = 2^{W-j-1}$.
For each box $b \in B_j$ one defines a subset $\bar{b} \subset B_j$, such that $\bar{b}$ only contains boxes that do not border $b$ and that did not interact with $b$ at previous levels (see \cref{fig}).
We let $E_b \subset E$ be the subgraph containing orbital pairs $(p,q)$ for which $(r_P + r_Q)/2 \in b$ ($E_{\bar{b}}$ is defined analogously).
To guarantee that orbital pairs in $b$ and $\bar{b}$ do not overlap, we stop at level $J$ such that $d_J \geq 8 r_{\max}$.
The result of the hierarchical partitioning is that $h_2^{\mathrm{long}}$ can be decomposed into operators $h_b$ describing interactions between orbital pairs in $b$ and its interaction set $\bar{b}$,
\begin{align}\label{eq:h_b}
	h_b = \frac12\;\;\sum_{\mathclap{\substack{(p,q)\in E_b\\ (r,s)\in E_{\bar{b}}}}} h_{pqrs} \T{p}{q}\T{r}{s}\quad\; h_2^{\text{long}} = \sum_{j = 1}^J \sum_{b\in B_j} h_b.
\end{align}
Next, a multipole expansion of the long-range Coulomb potential into real spherical harmonics is performed, which results in the standard low-rank decomposition of FMM methods (see End Matter for details).

\begin{lemma}\label[lemma]{lemma:hb multipole}
	Truncating the multipole expansion of each $h_b$ to $h_b^{(\degexp)} = \frac12\sum_{\ell=0}^{\degexp}\sum_{m=-\ell}^{\ell} v_b^{\ell m} w_b^{\ell m}$ such that the error in $h_2^{\mathrm{long}}$ is less than $\eps$ requires $\degexp = \bigO\big(\log\frac{\nsys}{\eps}\big)$.
\end{lemma}  

Here, $v_b^{\ell m}$ and $w_b^{\ell m}$ are 1-electron operators acting on the qubits of $b$ and $\bar{b}$, respectively.
The sparseness-structure of these operators is still given by the graphs $E_b$ and $E_{\bar{b}}$, which have degree at most $\Delta$.
By Vizing's theorem on edge coloring \cite{balakrishnan_textbook_2012} we can decompose these graphs into at most $\Delta + 2$ matchings (graphs of maximum degree 1).
Hence, each $v_b^{\ell m}$ and $w_b^{\ell m}$ can be decomposed into $(\Delta +2)$ 1-sparse operators (\textit{i.e.} their coefficient matrix has at most one non-zero entry per row and column).

\textit{Trotterized simulation} ---
In the Jordan--Wigner encoding \cite{jordan_uber_1928} the fermionic operators are mapped to qubit operators according to
\begin{align*}
	\T{p}{q} &\to \frac{1}{2}\big(X_p \vec{Z}_{pq} X_{q} + Y_p \vec{Z}_{pq} Y_q\big) \quad\! T_{pp}\to\frac{1}{2}\left(I - Z_p\right),
\end{align*}
where $\vec{Z}_{pq} = Z_{p+1} \cdots Z_{q-1}$.
An occupation number vector in the Fock space is mapped to the computational basis state with the same bitstring representation.
The mapping depends on the ordering of the orbital indices $p \in [\nsys]$, which may be chosen arbitrarily.
Fermionic permutations (FP) are operations on fermionic states encoded in qubits that change the ordering of the fermionic modes.
Such permutations on $\nsys$ qubits can be performed in depth $d_{\mathrm{FP}} = \bigO(\log^2\nsys)$ without ancillas \cite{constantinides_low-depth_2025}, or depth $\bigO(\log\nsys)$ with $\Theta(\nsys)$ ancillas \cite{maskara_fast_2025}. 

\begin{figure}
	\centering
	\begin{subfigure}{\linewidth}
		\centering
		\includegraphics[width=\linewidth]{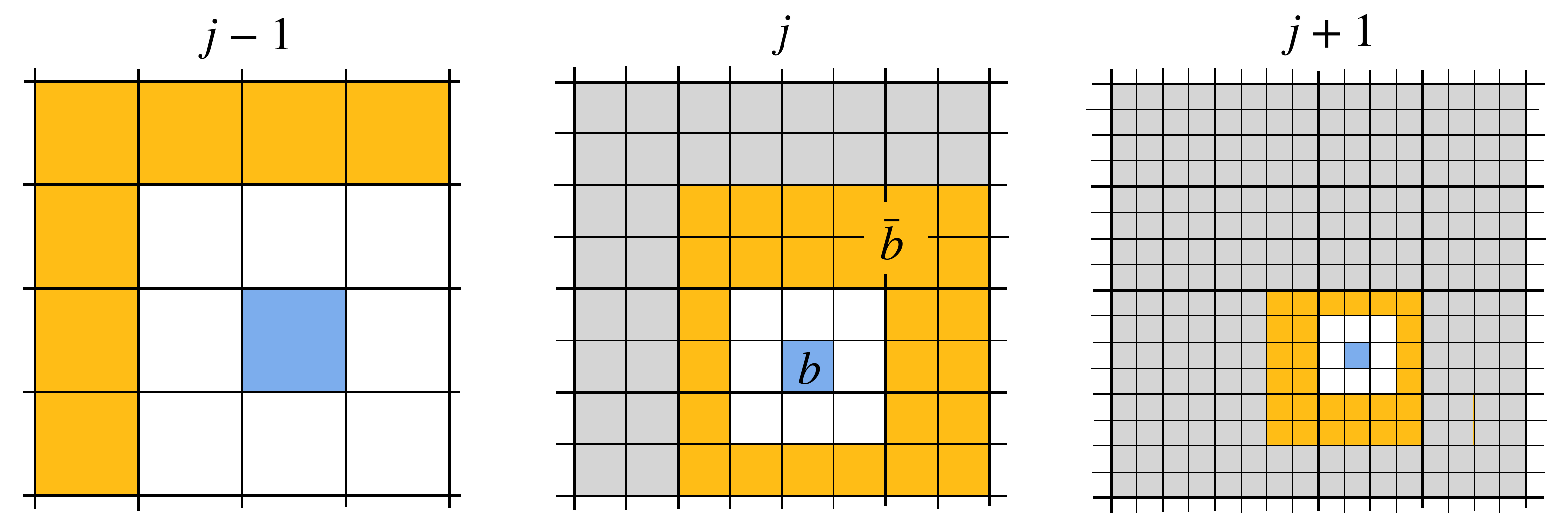}
	\end{subfigure}
	
	\vspace{0.2cm}
	\begin{subfigure}{\linewidth}
		\centering
		\includegraphics[width=\linewidth]{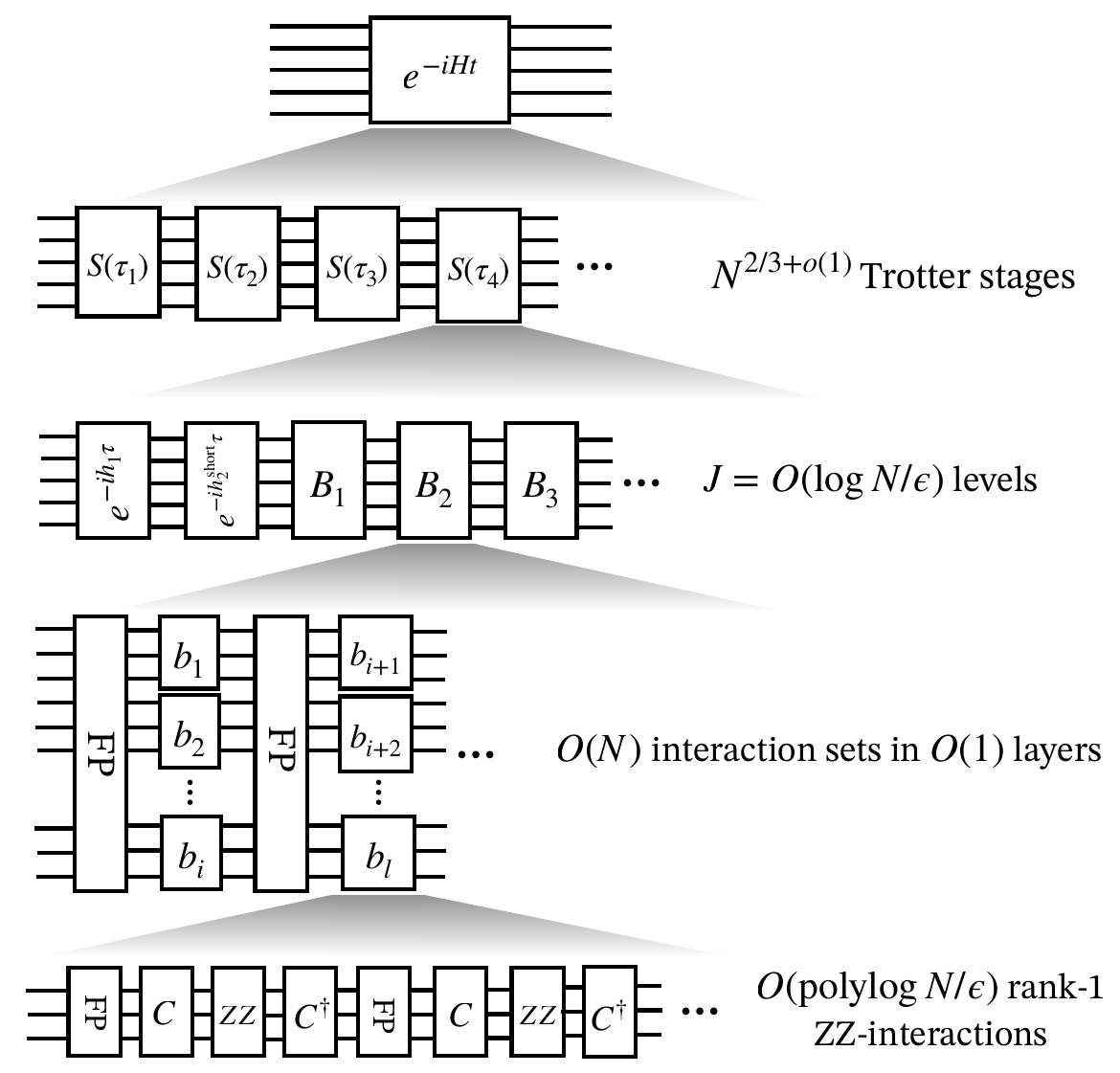}
	\end{subfigure}
	\caption{Top: Hierarchical decomposition in 2D. Bottom: Structure of the simulation algorithm.}
	\label{fig}
\end{figure}

Since an interaction $h_b$ only involves a constant number of nearby boxes in $B_j$, the layout of interactions at a given level $j$ can be parallelized to a constant number of layers.
We apply a fermionic permutation before each layer such that for all $h_b$ in the layer, the orbitals taking part in orbital pairs in $b$ appear in a sequence without gaps in the fermionic ordering (and the same for $\bar{b}$).
This way, all fermionic operations on orbitals in $h_b$ only involve qubits belonging to $b$ and $\bar{b}$.
Let $v$ and $w$ be 1-sparse 1-electron operators resulting from the decomposition discussed before.
The evolution $e^{-ivw t}$ serves as the building block for Trotterizing the long-range interactions.
Within $v$ and $w$, each orbital is interacting with only one other orbital, or with itself.
With a fermionic permutation all such pairs $(p,q)$ in $b$ can be put in consecutive order, $q=p+1$. 
In the Jordan--Wigner encoding, $v$ takes the form
\begin{align*}
	\sum_{p< q} \frac{v_{pq}}{2} (X_p X_q + Y_p Y_q) + \sum_p \frac{v_{pp}}{2} (I - Z_p).
\end{align*}
Since $C(X_pX_q + Y_pY_q) C^{\dag} = Z_p - Z_q$ for a Clifford gate $C$, $v$ can be mapped to $\alpha_0I + \sum_p \alpha_p Z_p$ for some $\alpha_p \in \RR$ (and similar for $w$), and we have 
\begin{align*}
	vw \to \sum_{p\in b} \alpha_p\beta_0 Z_p + \sum_{q\in \bar{b}} \alpha_0\beta_q Z_q + \sum_{p\in b,q\in\bar{b}} \alpha_p \beta_q Z_pZ_q.
\end{align*}
Here we ignored the identity term, and the summation over $p\in b$ indicates all orbitals that are involved with at least one orbital pair $(p,q)$ in $b$.
The first two terms are implemented by one layer of single-qubit rotations.
We are left with the $ZZ$-interaction with coefficient matrix of rank one.
Previous work states that a Trotter step for time $t$ of such a rank-one $ZZ$-interaction on $\nsys$ qubits can be implemented with $\bigO(\nsys \polylog(\nsys t))$ gates and $\bigO\big(\log\frac{\nsys t}{\eps}\big)$ ancillas \cite{low_complexity_2023}.
An explicit circuit is not described in that work, but conventional quantum arithmetic techniques achieve this gate count.
However, no claim about the circuit depth is made.
We close this gap and show that a rank-one $ZZ$-interaction can be implemented in polylogarithmic depth in the number of qubits:
\begin{lemma}\label[lemma]{lemma:rank1 dirty ancillas}	
	Consider operators $v = \sum_{p\in[N_b]} \alpha_pZ_p$, $w = \sum_{q\in [N_{\bar{b}}]} \beta_q Z_{N_b + q}$ such that $\abs*{\alpha_p\beta_q} < \mathrm{const.}$ and $N_b,N_{\bar{b}} = \bigO(\nsys)$.
	Then $e^{-ivwt}$ can be implemented to precision $\eps$ in depth $d_{ZZ} = \bigO(\log\frac{\nsys t}{\eps}\log\nsys)$ using $\bigO(\log\frac{\nsys t}{\eps})$ dirty ancilla qubits.
\end{lemma}

The proof and the quantum circuit behind it are detailed in the Appendix.
By reducing the number of parallel interactions $h_b$ per layer by a constant factor, we have sufficiently many idle system qubits, so that no additional qubits are required. 
There are at most $(\degexp+1)^2(\Delta+2)^2$ such rank-1 interactions per $h_b$, and the number of layers per level $j$ is bounded by a constant, hence the total depth for a Trotter stage of $h_{2}^{\mathrm{long}}$ is $\bigO(J\degexp^2\Delta^2 (d_{ZZ} + d_{\mathrm{FP}}))$.
The remaining terms, $h_1$ and $h_2^{\mathrm{short}}$, each contain only $\tilde\bigO(\nsys)$ terms.
Their Trotter stages can be implemented in polylogarithmic depth using an analogous decomposition into 1-sparse 1-electron interactions, skipping the multipole expansion step.
We arrive at a Trotter stage implementation with significantly lower depth than previous methods.

\begin{theorem}\label[theorem]{thm:trotter step}
	For $H\in\mathrm{MESH}(\zeta,\beta,\localsize,\delta,A)$ with $\nsys$ orbitals one can implement a Trotter stage for time $t$ to precision $\eps>0$ in depth $\bigO\big(\polylog\frac{\nsys t}{\eps}\big)$.
\end{theorem}

Higher-order Trotter product formulas $S_{k}(t)$ are constructed by sequences of Trotter stages with varying time steps.
To analyze the Trotter error, we rely on a result by Childs et al. \cite{childs_theory_2021} that gives an upper bound based on the commutator structure of the Hamiltonian terms.

\begin{theorem}\label[theorem]{thm:trotter error}
	Let $H\in\mathrm{MESH}(\zeta,\beta,\localsize,\delta,A)$ with $\nsys$ orbitals, then the error $\norm{e^{-iHt} - S_{k}(t)}$ for the decomposition of $H$ underlying \Cref{thm:trotter step} is at most
	\begin{align*}
		\bigO\big(8^k(k+1)!t^{k+1}\nsys\log^{k+1}\nsys\big) \quad&\text{for}\quad d=1\\
		\bigO\big(8^k(k+1)!t^{k+1}\nsys^{(k+1)(\frac{d-1}{d}) + 1}\big)\quad &\text{for}\quad d>1.
	\end{align*}
\end{theorem}

Approximating $e^{-iHt}$ by a sequence of $t/\tau$ Trotter steps $S_{k}(\tau)$ of short step sizes $\tau$ allows us to reduce the Trotterization error arbitrarily \cite{childs_theory_2021}.
The error decreases faster for higher-order formulas, but the increase in circuit complexity per step is exponential in $k$.
The optimal Trotter order is a function of $\nsys, t$ and $\eps$, and together with \cref{thm:trotter step} this yields the main result of our work:
\begin{theorem}\label[theorem]{thm:main result}
	Let $H\in\mathrm{MESH}(\zeta,\beta,\localsize,\delta,A)$ with $\nsys$ orbitals, then $e^{-iHt}$ can be implemented to precision $\eps>0$ with $\nsys^{\frac{2d-1}{d}}t\left(\nsys t/\eps\right)^{\smallo(1)}$ gates in depth $\nsys^{\frac{d-1}{d}}t(\nsys t/\eps)^{\smallo(1)}$ and without ancillas.
\end{theorem}

\textit{Discussion} ---
In this work we formalized the molecular electronic-structure setting based on atomic orbitals, and we showed existence of a localized orthonormal basis under realistic conditions.
We presented Hamiltonian decompositions and quantum circuits that result in a Trotter time-evolution algorithm that, compared to all previous works using atomic orbitals, is cheaper in terms of total gate count, gate depth and the number of ancillas.
Similar to extensions of the FMM, our polylogarithmic-depth Trotter step implementation could be applied to other types of 2-body interactions between fermions or spins that are described by a well-behaved potential.
We have focused on the asymptotic cost scaling and did not investigate at what system size the improved scaling outweighs the overhead incurred by our algorithm. 
Such a study requires optimizations of constants, logarithmic factors and circuit constructions and is left for future work.
Our scaling of gate count and depth is determined by fundamental quantities of the problem: the dimension and the exponent of decay of the Coulomb potential.
An open question is whether further improvement is possible.  
Lastly, we note that sublinear-depth simulation with current qubitization methods is obstructed by their query depth $\Omega(\lambda t)$, where $\lambda$ is the normalization factor of the Hamiltonian obeying $\lambda \geq \norm{H}=\Omega(\nsys)$.

\textit{Acknowledgements} --- 
This work was funded by NSF grant PHY-2325080, the MIT-IBM Watson AI Lab, and the U.S. Department of Energy, Office of Science, National Quantum Information Science Research Centers, Quantum Systems Accelerator (Award No. DE-SCL0000121).

\bibliography{references}

\section{End Matter}\label{sec:end matter}
Because of the exponential decay property an atomic orbital $\ao{Ii}$ is Lipschitz-continuous, and together with the normalization condition this implies $\abs{\ao{Ii}(r)} \to 0$ for $\abs{r} \to \infty$.
By the gradient theorem, we thus have
\begin{align*}
	\ao{Ii}(r) = -\int_{0}^{\infty} \nabla \ao{Ii}\left(r + s\frac{r - r_I}{\abs{r - r_I}}\right)\cdot\frac{r-r_I}{\abs{r-r_I}}ds
\end{align*}
and therefore the decay on $\ao{Ii}$ itself,
\begin{align*}
	\abs{\ao{Ii}(r)} \leq \beta \zeta e^{-\zeta\abs{r - r_I}}\int_0^{\infty}e^{-\zeta s}ds = \beta e^{-\zeta\abs{r-r_I}}.
\end{align*}

The off-diagonal entries of the Gram matrix $S$ are bounded by
\begin{align*}
	\abs{S_{Ii,Jj}} &\leq \beta^2 \int e^{-\zeta\abs{r_I - r}} e^{-\zeta\abs{r_J - r}}dr\\
	 &\leq  \frac{8\pi\beta^2}{\zeta^3}e^{-\zeta \abs{r_I - r_J}/2}\\
\end{align*}
Assuming atomic orbitals are on-site orthogonal, we can use \cref{lemma:aux} and Gershgorin's theorem to bound the minimal eigenvalue of the Gram matrix $S$ by
\begin{align*}
	\lambda_{\min}(S) &\geq \min_{Ii} \left(1 - \sum_{Jj\neq Ii} \abs{S_{Ii,Jj}} \right)\\
	&\geq \min_{I}\left( 1 - \frac{8\pi \beta^2 \localsize}{\zeta^3} \sum_{J\neq I} e^{-\zeta\abs{r_I - r_J}/2} \right)\\
	&\geq 1 - \frac{16\pi \beta^2 \localsize}{\zeta^3} d! \left(\frac{12}{\delta\zeta}\right)^d e^{-\zeta \delta/4}.
\end{align*}

The truncation of orbitals to radius $\rupper$ amounts to a perturbation of the coefficients of $H$, which we bound now.
We obtain truncated orbitals from the original ones by multiplying by a smooth cutoff function, $\orb{p}'(r) = \orb{p}(r)\eta(\abs{r-r_P}/\rupper)$, where $\eta=1$ on $[0,1/2]$ and $\eta=0$ on $[1,\infty)$.
The tail $\orb{p}^{\mathrm{tail}} = \orb{p} - \orb{p}'$ decays as $e^{-\zeta\rupper/4}e^{-\zeta\abs{r-r_P}/2}$.
We expand $\orb{p} = \orb{p}' + \orb{p}^{\mathrm{tail}}$ in $h_{pq}$ and $h_{pqrs}$ and truncate all terms containing a tail.
The sum of truncated terms is $\bigO(e^{-\zeta\rupper/4}\nsys^{\frac{2d-1}{d}})$, where we used \Cref{lemma:mesh norm}, giving $\rupper = \bigO\left(\zeta^{-1} \log\frac{\nsys}{\eps}\right)$ for truncation error $\eps$.
 This threshold determines the maximal degree $\Delta$ of the orbital overlap graph $E$; the number of spheres of radius $\delta$ fitting into a sphere of radius $2\rupper$ is at most $(2\rupper/\delta)^d$, so we have $\Delta = \bigO\left((\frac{1}{\zeta\delta}\log\frac{\nsys}{\eps})^d\right)$.

Next, we provide details on how the number of levels $J$ of the hierarchical decomposition should be chosen.
Let $d_j$ be the box width at level $j$.
A homogeneous distribution of atoms gives $J = \bigO(\log \nsys)$, but if the system was composed of two subsystems at distance $\bigO(2^{\nsys})$ a naive bound would yield $J=\bigO(\nsys)$. 
However, we can truncate interactions between such subsystems; by \cref{lemma:coeff decay} a long-range term will be smaller than $\eps$ if it involves orbital pairs further apart than a distance $\bigO(\eps^{-1})$.
There are $\tilde\bigO\left(\nsys^2\right)$ long-range terms, so truncating all terms involving atomic orbitals further apart than $\tilde\bigO(\nsys^2\eps^{-1})$ leads to error at most $\tilde\bigO\left(\eps\right)$.
Consequently, the maximal relevant extent of the system is $\tilde\bigO\left(\nsys^3\eps^{-1}\right)$ and hence $J = \tilde\bigO(\log\frac{\nsys}{\eps})$.
Non-interacting subsystems can be treated independently, in particular our simulation algorithm can run on both subsystems in parallel.
At the finest level, the box width needs to satisfy $d_J > 2 \rupper$, as otherwise orbitals in $b$ and $\bar{b}$ can overlap.
We set $J$ to be the largest integer such that $d_J>8\rupper$, so the supports of interacting orbital pairs are guaranteed to be at least $3 d_j/4$ apart for any $j\in [J]$.

The multipole expansion for the Coulomb potential is 
\begin{align}\label{eq:multipole expansion infty}
	\frac{1}{\abs{r_1 - r_2}}\! =\! \sum_{\ell = 0}^{\infty}\; \sum_{\mathclap{m=-\ell}}^{\ell}\! \abs{r_1}^{\ell} Y_{\ell m}(\theta_1,\varphi_1)\!\frac{Y_{\ell m}(\theta_2,\varphi_2)}{\abs{r_2}^{\ell+1}}
\end{align}
where $Y_{\ell m}(\theta,\varphi)$ are real spherical harmonics (in Racah's normalization), and $(\abs{r}, \theta,\varphi)$ are spherical coordinates for $r$.
We expand around the center of $b$, so that $\abs{r_1} \leq \sqrt{3}d_j/2 + r_{\max} \leq d_j$ and $r_2 \geq 3 d_j/2 - r_{\max} \geq 11d_j/8$ at level $j$, where $d_j$ is the box width.
Truncating the expansion to $\degexp<\infty$ yields an approximation with a pointwise error of at most $\abs{\abs{r_1} - \abs{r_2}}^{-1}(\abs{r_1}/\abs{r_2})^{\degexp + 1} \leq 3(8/11)^{\degexp}/d_j$ \cite{greengard_new_1997}.
We define multipole moments for $b$ and $\bar{b}$:
\begin{align}\label{eq:defn multipole moments}
	v_{pq}^{\ell m}:= \int \orb{p}^*(r_1) \abs{r_1}^{\ell} Y_{\ell m}(\theta_1,\varphi_1) \orb{q}(r_1)dr_1\\
	w_{rs}^{\ell m}:= \int \orb{r}^*(r_2) \abs{r_2}^{-\ell-1} Y_{\ell m}(\theta_2,\varphi_2) \orb{s}(r_2) dr_2,
\end{align}
then the multipole operators of \cref{lemma:hb multipole} are
\begin{align}\label{eq:defn multipole ops}
	v_b^{\ell m}&:= \sum_{(p,q)\in E_b} v_{pq}^{\ell m} T_{pq}\\
	w_b^{\ell m}&:= \sum_{(r,s)\in E_{\bar{b}}} w_{rs}^{\ell m} T_{rs}.
\end{align}
With $\abs{Y_{\ell m}}\leq 1$ and $\int\abs{\orb{p}(r)\orb{q}(r)}dr\leq 1$ due to the Cauchy-Schwartz inequality we have $\abs{v_{pq}^{\ell m} w_{rs}^{\ell m}}\leq (8/11)^{\ell+1}/d_j \leq d_J = \bigO(1)$, hence the assumption on the coefficients of \Cref{lemma:rank1 dirty ancillas} is satisfied.
As explained in the main text, we can decompose further into 1-sparse 1-electron operators, which we denote here by $v_{b}^{\ell m \sigma}$ and $w_{b}^{\ell m \rho}$,
\begin{align}
	v_b^{\ell m} = \sum_{\sigma=1}^{\Delta + 2} v_{b}^{\ell m \sigma}\qquad
	w_b^{\ell m} = \sum_{\rho=1}^{\Delta + 2} w_{b}^{\ell m \rho}.
\end{align}

The decomposition of the Hamiltonian that is employed for the Trotter-step implementation reads
\begin{align}
	&H \approx \sum_{pq} h_{pq} T_{pq} + \frac12\sum_{pqrs} h_{pqrs}^{\mathrm{short}} T_{pq} T_{rs} \label{eq:hamil decomp}\\
	&+ \frac12 \sum_{j=1}^J \sum_{b\in B_j} \sum_{\ell=0}^{\degexp} \sum_{m=-\ell}^{\ell} \sum_{\sigma,\rho = 1}^{\Delta + 2} \sum_{\substack{(p,q)\in E_b\\(r,s) \in E_{\bar{b}}}} v_{pq}^{\ell m \sigma}  w_{rs}^{\ell m \rho} T_{pq} T_{rs}. \nonumber
\end{align}

The implementation of the long-range terms has been discussed in the main text; we explain the implementation of the short-range terms here.
Since $d_J \sim \Theta(\rupper)$, each remaining term in $h_2^{\mathrm{short}}$ acts locally in a neighborhood of size $\bigO(\rupper)$.
There are $\bigO(\nsys \rupper^{3d})$ such terms, and they can be parallelized in $\bigO(\rupper^{3d})$ layers, interleaved with fermionic permutations, resulting in polylogarithmic depth.
Lastly, $h_{1}$ is decomposed into $\Delta + 2$ 1-sparse 1-electron operators, and they are sequentially implemented by mapping each to a sum of single-qubit $Z$ operations as described in the main text. 

\onecolumngrid
\newpage
\section{Appendix}

We show related auxiliary lemmas concerning the Hamiltonian coefficients here.

\begin{lemma}\label[lemma]{lemma:aux}
	Let $\{r_1,r_2,\dots,r_{\nnuc}\}\subset \RR^3$ be in a d-dimensional configuration with $\min_{I\neq J}\abs{r_I - r_J}\geq \delta$, then
	\begin{align}
		&\max_I \sum_{\mathclap{\substack{J=1\\ \abs{r_I - r_J} \geq r'}}}^{\nnuc} e^{-\zeta\abs{r_I - r_J}} \leq 2d! \left(\frac{6}{\delta\zeta}\right)^d e^{-\zeta r'/2}  \qquad \forall \zeta > 0, \forall r' \geq 0\label{eq:sum exp} \\ 
		&\max_I \sum_{J\neq I}^{\nnuc} \frac{1}{\abs{r_I - r_J}}  = 
		\begin{cases}
			\bigO(\log\nnuc), \quad \text{if} \quad d=1\\
			\bigO(\nnuc^{\frac{d-1}{d}}), \quad \text{if} \quad d>1
		\end{cases}. \label{eq:sum coulomb}
	\end{align}
\end{lemma}
\begin{proof}
	The number of spheres of radius $\delta/2$ whose centers fit within a sphere of radius $r + \delta/2$ is at most $((r+\delta/2)/(\delta/2))^d \leq (3r/\delta)^d$.
	Let $r_I$ be at the origin, then
	\begin{align*}
		\sum_{\substack{J=1\\ \abs{r_I - r_J} \geq r'}}^{\nnuc} e^{-\zeta\abs{r_I - r_J}} &= \sum_{\substack{J=1\\ \abs{r_J} \geq r'}}^{\nnuc} \int_0^{\infty} \delta(r-\abs{r_J})e^{-\zeta r} dr\\
		&= \zeta \int_0^{\infty} \sum_{\substack{J=1\\ \abs{r_J} \geq r'}}^{\nnuc} \theta(r - \abs{r_J}) e^{-\zeta r}dr\\
		&\leq \zeta e^{-\zeta r'/2} \int_0^{\infty} \left(\frac{3r}{\delta}\right)^d e^{-\zeta r/2} dr\\
		&= 2d! \left(\frac{6}{\delta\zeta}\right)^d e^{-\zeta r'/2} 
	\end{align*}
	For \cref{eq:sum coulomb} we maximize by packing $\nnuc$ nuclei inside the $d$-dimensional ball of radius $\bigO(\nnuc^{1/d}\delta)$.
	Let $r_I = 0$, then by using similar steps as before, we find
	\begin{align*}
		\max_I\sum_{J\neq I}^{\nnuc} \frac{1}{\abs{r_I - r_J}} = \sum_{J} \int_{\delta}^{\bigO(\nnuc^{1/d}\delta)} \delta(r-\abs{r_J})\frac{1}{r}dr &\leq \frac{3^d}{\delta^d}\int_{\delta}^{\bigO(\nnuc^{1/d}\delta)} r^{d-2} dr = \begin{cases}
		\bigO(\log\nnuc), \quad \text{if} \quad d=1\\
		\bigO(\nnuc^{\frac{d-1}{d}}), \quad \text{if} \quad d>1.
		\end{cases}.
	\end{align*} 
\end{proof}

The localization properties of $H \in \mathrm{MESH}(\zeta,\beta,\localsize,\delta,A)$ allow us to bound the Hamiltonian coefficients.

\begin{lemma}\label[lemma]{lemma:coeff decay}
	Let $h_{pq}$ and  $h_{pqrs}$ be the 1- and 2-electron coefficients of an atomic orbital set of $\nnuc$ atoms with $\localsize$ orbitals per site and decay parameters $\beta,\zeta>0$, and let the nuclei be arranged in a $d$-dimensional configuration, then
	\begin{align*}
		\abs{h_{pq}} &\leq c_1 e^{-\frac{\zeta}{2}\abs{r_P - r_Q}}\cdot
		\begin{cases}
			\log(\nnuc), \quad \text{if} \quad d=1 \\
			\nnuc^{\frac{d-1}{d}}, \quad \text{if} \quad d>1
		\end{cases}\\
		\abs{h_{pqrs}} &\leq c_2 \min(1,\abs{r_I - r_J}^{-1}) e^{-\frac{\zeta}{2}(\abs{r_P - r_Q} + \abs{r_R - r_S})}
	\end{align*}
	for $c_1,c_2>0$ independent of $\nnuc$, and with $I\in\{P,Q\}, J\in\{R,S\}$ maximizing $\abs{r_I - r_J}$.
\end{lemma}
\begin{proof}
	For the electron--electron Coulomb integral, we first use the bound on the absolute value of the atomic orbitals, and then apply the triangle inequality in the exponents to get
	\begin{align*}
		\abs{h_{pqrs}} &\leq \beta^4 \int\int\frac{e^{-\zeta\abs{r_1 - r_P}} e^{-\zeta\abs{r_1 - r_Q}} e^{-\zeta\abs{r_2 - r_R}} e^{-\zeta\abs{r_2 - r_S}}}{\abs{r_1 - r_2}}dr_2 dr_1\\
		&\leq \beta^4 e^{-\frac{\zeta}{2}(\abs{r_P - r_Q} + \abs{r_R - r_S})} \int\int\frac{e^{-\zeta\abs{r_1 - r_P}/2} e^{-\zeta\abs{r_1 - r_Q}/2} e^{-\zeta\abs{r_2 - r_R}/2} e^{-\zeta\abs{r_2 - r_S}/2}}{\abs{r_1 - r_2}}dr_2 dr_1\\
		&\leq \beta^4 e^{-\frac{\zeta}{2}(\abs{r_P - r_Q} + \abs{r_R - r_S})} \int\int\frac{e^{-\zeta\abs{r_1 - r_I}/2} e^{-\zeta\abs{r_2 - r_J}/2}}{\abs{r_1 - r_2}}dr_2 dr_1.
	\end{align*}
	Here, $I\in\{P,Q\}$ and $J\in\{R,S\}$ are chosen such that $\abs{r_I - r_J}$ is maximized.
	The integral in the last line is finite for any values of $r_I,r_J$, in particular for $r_I = r_J$.
	An analytic solution to the integral is given in \cite{roothaan_study_1951},
	\begin{align*}
		\frac{\pi^2 4^6}{r\zeta^6}\left[1 - e^{-\frac{\zeta r}{2}}\left(1 + \frac{11}{32}\zeta r + \frac{3}{64}\zeta^2 r^2 + \frac{1}{384}\zeta^3 r^3 \right)\right] \leq \min\left(\frac{\pi^2 4^6}{r\zeta^6}, \frac{706\pi^2}{\zeta^5}\right) \quad \text{with}\quad r =\abs{r_I - r_J}.
	\end{align*}
	This proves the claimed bound on $\abs{h_{pqrs}}$.
	The kinetic integral can be bounded by
	\begin{align*}
		\abs{t_{pq}} &= \frac12\abs{\int \nabla\orb{p}^*(r) \cdot \nabla\orb{q}(r)dr} \leq \frac{\beta^2 \zeta^2}{2} \int e^{-\zeta\abs{r-r_P}} e^{-\zeta\abs{r-r_Q}}dr \leq \frac{4\pi \beta^2}{\zeta} e^{-\zeta\abs{r_P - r_Q}/2},
	\end{align*}
	where we made use of the gradient decay property of atomic orbitals, and then applied the triangle inequality.
	For the electron--nucleus Coulomb interaction term we bound the exponentials with similar steps,
	\begin{align*}
		\abs{u_{pq}} &\leq \beta^2 \sum_{I=1}^{\nnuc}Z_I\int\frac{e^{-\zeta\abs{r - r_P}}e^{-\zeta\abs{r-r_Q}}}{\abs{r-r_I}}dr \leq \beta^2 e^{-\zeta\abs{r_P - r_Q}/2} \sum_{I=1}^{\nnuc}Z_I \int \frac{e^{-\zeta\abs{r-r_P}/2}}{\abs{r-r_I}}dr,
	\end{align*}
	where the integral is solved by \cite{roothaan_study_1951}
	\begin{align*}
		\frac{\pi 4^3}{r \zeta^3}\left[1 - e^{-\frac{\zeta r}{2}}\left(1 + \frac{\zeta r}{4}\right) \right] \leq \min\left(\frac{\pi 4^2}{\zeta^2},\frac{\pi 4^3}{r \zeta^3}\right) \quad \text{with}\quad r = \abs{r_P - r_I}.
	\end{align*}
	Applying \cref{lemma:aux} to bound the sum over nuclei yields
	\begin{align*}
		\abs{u_{pq}} \leq e^{-\zeta\abs{r_P-r_Q}/2}\frac{\pi 4^3 \beta^2Z_{\max}}{\zeta^3}\left(\frac{\zeta}{4} + \sum_{I\neq P}^{\nnuc} \frac{1}{\abs{r_P - r_I}} \right) \leq ce^{-\frac{\zeta}{2}\abs{r_P - r_Q}}\cdot
		\begin{cases}
			\log(\nnuc),\quad\text{if}\quad d=1,\\
			\nnuc^{\frac{d-1}{d}},\quad\text{if}\quad d>1
		\end{cases}
	\end{align*}
	for some $c>0$.
	Lastly, \cref{lemma:aux} and a similar sequence of steps as for the bound on $\abs{h_{pqrs}}$ gives $\sum_r\abs{h_{prrq}} = \bigO(e^{-\zeta\abs{r_P - r_Q}/2})$.
	Combining the upper bounds on $\abs{t_{pq}}, \abs{u_{pq}}$ and $\abs{h_{prrq}}$ yields the claimed bound.
\end{proof}

\begin{lemma}\label[lemma]{lemma:mesh norm}
	Let $h_{pq}$ and  $h_{pqrs}$ be the 1- and 2-electron coefficients of an atomic orbital set of $\nnuc$ atoms with $\localsize$ orbitals per site and decay parameters $\beta,\zeta>0$, and let the nuclei be arranged in a $d$-dimensional configuration, then for any $r'\geq 0$
	\begin{align*}
		\max_p\sum_{\mathclap{\substack{q\\\abs{r_P - r_Q}>r'}}} \abs{h_{pq}} &= e^{-\zeta r'/4}\cdot
		\begin{cases}
			\bigO(\log\nsys),\quad \text{if}\quad d=1,\\
			\bigO(\nsys^{\frac{d-1}{d}}),\quad \text{if}\quad d>1
		\end{cases}\\
		\max_p\sum_{\mathclap{\substack{qrs\\\abs{r_P - r_Q}>r'\\\abs{r_R - r_S}>r'}}} \abs{h_{pqrs}} &= e^{-\zeta r'/2}\cdot
		\begin{cases}
			\bigO(\log\nsys),\quad \text{if}\quad d=1,\\
			\bigO(\nsys^{\frac{d-1}{d}}),\quad \text{if}\quad d>1
		\end{cases}.
	\end{align*}
\end{lemma}

\begin{proof}
	Combining \cref{lemma:aux} for $r'=0$ with \cref{lemma:coeff decay} gives
	\begin{align*}
		\max_p\sum_{\mathclap{\substack{q\\\abs{r_P - r_Q}>r'}}} \abs{h_{pq}} \leq c_1 \localsize 2d! \left(\frac{6}{\zeta\delta}\right)^d e^{-\zeta r'/4} \cdot 
		\begin{cases}
			\log(\nnuc), \quad \text{if} \quad d=1 \\
			\nnuc^{\frac{d-1}{d}}, \quad \text{if} \quad d>1
		\end{cases}
	\end{align*}
	which proves the first part because $\nsys = \Theta(\nnuc)$.
	Similarly, from \cref{lemma:aux} and \cref{lemma:coeff decay} we obtain the bound on the 2-electron coefficients,
	\begin{align*}
		\max_p \sum_{qrs}\abs{h_{pqrs}} &\leq c_2\localsize^3\max_P \sum_{\mathclap{\substack{QRS\\\abs{r_P-r_Q}> r'\\\abs{r_R-r_S}>r'}}} \min\left(1,\min_{I\in\{P,Q\}, J\in \{R,S\}}\abs{r_I - r_J}^{-1}\right)  e^{-\frac{\zeta}{2}\abs{r_P - r_Q}} e^{-\frac{\zeta}{2}\abs{r_R - r_S}}\\
		&\leq c_2\localsize^3\max_P \left(\sum_{R\neq P} \abs{r_P - r_R}^{-1} + 1\right) \left(\sum_{\substack{Q\\\abs{r_P - r_Q}>r'}} e^{-\frac{\zeta}{2}\abs{r_P - r_Q}}\right) \left(\sum_{\substack{S\\\abs{r_R - r_S}>r'}} e^{-\frac{\zeta}{2}\abs{r_R - r_S}}\right)\\
		& = e^{-\zeta r'/2} \cdot 
		\begin{cases}
			\bigO(\log\nsys),\quad \text{if}\quad d=1,\\
			\bigO(\nsys^{\frac{d-1}{d}}),\quad \text{if}\quad d>1
		\end{cases}.
	\end{align*}
\end{proof}

\subsection{Proof of \Cref{thm:lowdin}}

We combine ideas of Jaffard \cite{jaffard_proprietes_1990}, who studied the algebra of infinite matrices whose entries have an exponential decay structure, and of Benzi \cite{benzi_localization_2016}, who establishes entry-wise decay bounds for functions of sparse or banded matrices by combining polynomial approximation with the linear independence bound $A>0$.
In contrast to \cite{jaffard_proprietes_1990}, we consider finite dimensional block-matrices, and in contrast to \cite{benzi_localization_2016}, we do not perform an intermediate truncation to obtain sparse matrices, so exponential decay is independent of a truncation threshold.  

The class of block matrices with exponential decay is defined as follows: 

\begin{definition}
	\label{defn:decay matrices}
	Let $\zeta,\beta,\delta>0$, and let $S$ be a block matrix with $\localsize\times\localsize$ blocks $S_{IJ}$ for $I,J \in \{1,2,\dots,\nnuc\}$.
	Then $S$ is an element of the decay class $\decayalgebra(\zeta,\beta,\delta)$ if there exists a set of points $\{r_1,r_2,\dots,r_{\nnuc}\}\subset \RR^d$ satisfying the separation condition $\min_{I\neq J}\abs{r_I - r_J}\geq \delta$ and 
	\begin{align*}
		\norm{S_{IJ}}_F \leq \localsize \beta e^{-\zeta\abs{r_I - r_J}},
	\end{align*}
	where $\norm{\cdot}_F$ denotes the Frobenius norm.
\end{definition}

The Frobenius-norm $\norm{A}_F = \left(\sum_{ij}\abs{a_{ij}}^2\right)^{1/2}$ is used here, because for the Gram matrix we obtain entry-wise bounds.
Taking powers of a decay matrix changes the decay parameters $\zeta$ and $\beta$, but the parameters stay independent of the number of blocks $\nnuc$.
In the following, we let $c(\zeta,\delta,d) = 2d!(6/(\delta\zeta))^d$ be the prefactor of the bound in \cref{eq:sum exp}.

\begin{lemma}\label[lemma]{lemma: decay matrix power}
	Let $S \in \decayalgebra(\zeta,\beta,\delta)$ with $\localsize\times\localsize$ blocks, and define $R(\zeta') = \localsize \beta c(\zeta-\zeta',\delta,d)$, then for any $\zeta'\in(0,\zeta)$ and all $k\geq 1$,
	\begin{align*}
		S^k \in \decayalgebra(\zeta',\beta R(\zeta')^{k-1},\delta).
	\end{align*}
\end{lemma}
\begin{proof}
	The proof is by induction over $k$.
	That the base case $k=1$ holds follows immediately from \Cref{defn:decay matrices} for any $\zeta'<\zeta$.
	Let $k>1$ and assume that $S^{k-1} \in \decayalgebra(\zeta',\beta R(\zeta')^{k-2},\delta)$, then we invoke \cref{lemma:aux} to see that
	\begin{align*}
		\norm{(S^k)_{IJ}}_F &\leq \sum_{M=1}^{\nnuc} \norm{S_{IM}}_F \norm{(S^{k-1})_{MJ}}_F\\
		&\leq \localsize^2 \beta^2 R(\zeta')^{k-2} \sum_{M=1}^{\nnuc}  e^{-\zeta\abs{r_I - r_M}} e^{-\zeta'\abs{r_M - r_J}}\\
		&\leq \localsize^2 \beta^2 R(\zeta')^{k-2} \sum_{M=1}^{\nnuc} e^{-(\zeta-\zeta')\abs{r_I - r_M}} e^{-\zeta'(\abs{r_I - r_M} + \abs{r_M - r_J})}\\
		&\leq \localsize^2 \beta^2 R(\zeta')^{k-2} c(\zeta - \zeta',\delta,d) e^{-\zeta'\abs{r_I - r_J}}\\
		& = \localsize \beta R(\zeta')^{k-1} e^{-\zeta'\abs{r_I - r_J}}.
	\end{align*}
\end{proof}

To bound the decay on $S^{-1/2}$, we look at its polynomial approximation, and we rely on standard convergence bounds using Chebyshev polynomials \cite{trefethen_approximation_2013}.

\begin{lemma}\label[lemma]{lemma:poly approx}
	Let $0 < A \leq B$, and let $p_l$ be the degree-$l$ Chebyshev truncation for $f(\lambda) = \lambda^{-1/2}$ on $[A,B]$.
	Then, for $C>0$ only depending on $A$ and $B$,
	\begin{equation}
		\sup_{\lambda \in [A, B]} \abs{f(\lambda) - p_l(\lambda)} \;\leq\; C\, \rho^l, \quad \text{where} \quad \rho = \frac{\sqrt{B/A}-1}{\sqrt{B/A}+1} < 1.
	\end{equation}
\end{lemma}
We emphasize that we are not merely interested in showing that $S^{-1/2}$ is a member of a decay class for some decay parameters.
Instead, we show the stronger statement that the decay parameters of $S^{-1/2}$ are independent of $\nnuc$.

\begin{lemma}\label[lemma]{lemma:decay s^-1/2}
	Let $S \in \decayalgebra(\zeta,\beta,\delta)$ with $\localsize\times\localsize$ blocks be strictly positive definite with spectrum inside $[A,B]$ for $B \geq A > 0$.
	Then $S^{-1/2} \in \decayalgebra(\tilde\zeta,\tilde\beta,\delta)$ for some decay parameters $\tilde\zeta,\tilde\beta$ independent of $\nnuc = dim(S)/\localsize$.
\end{lemma}
\begin{proof}
	Let $p_l(\lambda) = \sum_{k=0}^l a_k T_k(u(\lambda))$ be the degree-$l$ Chebyshev truncation for $f(\lambda) = \lambda^{-1/2}$ on $[A,B]$, where $u(\lambda) = (2\lambda - (A+B))/(B -A)$ maps $[A,B]$ to $[-1,1]$.
	The sum of absolute monomial coefficients of $T_k$ is at most $\nu^k$ with $\nu = (1 + \sqrt{2})$.
	For $\zeta' \in (0,\zeta)$ define $R' = \nu\max(1,R(\zeta'))$, where $R(\zeta')$ is from \cref{lemma: decay matrix power}, then
	\begin{align*}
		\norm{(p_l(S))_{IJ}}_F \leq \sum_{k=0}^l \abs{a_k} \norm{(T_k(u(S)))_{IJ}}_F \leq \frac{2\norm{a}_1 \localsize\beta}{B-A} R'^l e^{-\zeta'\abs{r_I - r_J}}.
	\end{align*}
	Here, \Cref{lemma: decay matrix power} was applied to $u(S)$, which changed $\beta \to \beta 2/(B-A)$.
	The Chebyshev coefficients decay exponentially for the analytic function $f(\lambda)$ on $[A,B]$, therefore the 1-norm of the Chebyshev coefficients, $\norm*{a}_1 = \sum_{k=0}^{\infty}\abs*{a_k}$, is finite. 
	As we assume that $A$ and $B$ are independent of $\nnuc$, it is independent of $\nnuc$ as well.
	The existence of $S^{-1/2}$ follows directly from the assumption on the spectrum of $S$ and the spectral theorem.
	By \cref{lemma:poly approx} we can bound 
	\begin{align*}
		\norm{(S^{-1/2} - p_l(S))_{IJ}}_F &\leq \sqrt{\localsize} \norm{S^{-1/2} - p_l(S)}_2\leq \sqrt{\localsize} \sup_{\lambda\in[A,B]} \abs{\lambda^{-1/2} - p_l(\lambda)} \leq C \sqrt{\localsize} \rho^l.
	\end{align*}
	We set $l = \lfloor b\abs{r_I - r_J}\rfloor$ for some constant $b>0$ and apply the triangle inequality to obtain
	\begin{align*}
		\norm{(S^{-1/2})_{IJ}}_F &\leq \frac{2\norm*{a}_1 \localsize \beta }{B-A}R'^{b\abs{r_I - r_J}} e^{-\zeta'\abs{r_I - r_J}} + C \sqrt{\localsize} \rho^{b\abs{r_I - r_J}} \\
		&\leq  \frac{2\norm{a}_1\localsize \beta}{B-A} e^{-(\zeta' - b\log(R'))\abs{r_I - r_J}} + C \sqrt{\localsize} e^{-b\abs{\log\rho}\abs{r_I - r_J}}.
	\end{align*}
	The right-hand side is independent of $\nnuc$, and we can find the optimal decay rate by minimizing the maximal exponent with respect to $b$: $\tilde\zeta(\zeta') = \zeta'\abs{\log\rho}/(\log R' + \abs{\log\rho}) > 0$.
	Thus, with $\tilde\beta = \frac{2\norm{a}_1 \beta}{B-A} + C/\sqrt{\localsize}$, we reach the desired result
	\begin{align*}
		\norm{(S^{-1/2})_{IJ}}_F & \leq \localsize\tilde\beta e^{-\tilde\zeta\abs{r_I - r_J}}.
	\end{align*}
\end{proof}

We conclude the proof of \cref{thm:lowdin} now.
As shown in the End Matter, the entries of the Gram matrix of an atomic orbital set decay as $\abs{S_{Ii,Jj}} = \bigO\left(e^{-\zeta\abs{r_I - r_J}/2}\right)$, and therefore
 \begin{align*}
	\norm{S_{IJ}}_F \leq \localsize \beta' e^{-\zeta\abs{r_I - r_J}/2}
\end{align*}
for some $\beta'> 0$.
Therefore, $S$ is a member of the decay class $\decayalgebra(\zeta/2,\beta',\delta)$ with decay parameters independent of $\nnuc$.
Using \cref{lemma:aux}, the largest eigenvalue of $S$ can be bounded independently of $\nnuc$ by
\begin{align*}
		\lambda_{\max} &= \norm{S}_2 \leq \sqrt{\norm{S}_1 \norm{S}_{\infty}} \leq \max_{I,i} \sum_{J,j}\abs{S_{Ii,Jj}} \leq \localsize\beta' \max_I \sum_J e^{-\zeta\abs{r_I - r_J}/2} \leq \localsize \beta' d! \left(\frac{3}{\delta\zeta}\right)^d,
\end{align*}
so the spectral range of $S$ is $\nnuc$-independent.
 Consequently, $S^{-1/2}$ is also a decay matrix with decay parameters independent of $\nnuc$ (\cref{lemma:decay s^-1/2}).
Finally, we need to show that the Löwdin orbitals satisfy the atomic orbital decay condition with decay parameters independent of $\nnuc$.
We set $\zeta'' = \min(\tilde\zeta,\zeta)/2$ and use the triangle inequality in the exponent to obtain the upper bound with appropriately defined decay parameters
\begin{align*}
	\abs{\nabla\orb{Ii}(r)} &\leq \sum_{J,j} \abs{\left(S^{-1/2}\right)_{Ii,Jj}} \abs{\nabla\ao{Jj}(r)} \leq \localsize^{3/2} \beta'\beta\zeta \sum_{J}  e^{-\tilde\zeta\abs{r_I - r_J}} e^{-\zeta\abs{r - r_J}} \leq \localsize^{3/2} \beta' \beta\zeta c(\zeta'',\delta,d)e^{-\zeta''\abs{r-r_I}}.
\end{align*}

\subsection{Proof of \Cref{lemma:hb multipole}}

In this part of the Appendix, we bound the claimed error due to truncation of the multipole expansion.
We insert the definitions of the multipole expansion and multipole operators, \cref{eq:multipole expansion infty,eq:defn multipole moments,eq:defn multipole ops}, and use $\norm{T_{pq}}=1$, to obtain
\begin{align*}
	&\norm{h_b - h_b^{(\degexp)}} \leq \frac12\sum_{\substack{(p,q)\in E_b\\(r,s)\in E_{\bar{b}}}} \abs{h_{pqrs} - \sum_{\ell=0}^{\degexp} \sum_{m=-\ell}^{\ell} v_{pq}^{\ell m} w_{rs}^{\ell m} }\\
	&\leq \frac12\sum_{\substack{(p,q)\in E_b\\(r,s)\in E_{\bar{b}}}} \int \int \abs{\orb{p}^*(r_1)\orb{q}(r_1)} \abs{\frac{1}{\abs{r_1 - r_2}} - \sum_{\ell=0}^{\degexp} \sum_{m=-\ell}^{\ell} \abs{r_1}^{\ell} Y_{\ell m}(\theta_1,\varphi_1)\frac{Y_{\ell m}(\theta_2,\varphi_2)}{\abs{r_2}^{\ell+1}}} \abs{\orb{r}^*(r_2)\orb{s}(r_2)} dr_1 dr_2.
\end{align*}
 At level $j$, the maximal distance between the supports of $\orb{p}\orb{q}$ and $\orb{r}\orb{s}$ is $\abs{r_1 - r_2} \leq \sqrt{3 \cdot 4^2} d_j + 2\rupper<8d_j$, where we ignored modifications due to $r_{\max}$.
 The difference in the integrand is the truncation error in the multipole expansion (see End Matter), which is upper bounded by $3(8/11)^{\degexp}/d_j \leq 24\cdot e^{-\degexp/4}/\abs{r_1 - r_2}$, and consequently
 \begin{align*}
 	\norm{h_2^{\mathrm{long}} - \frac12\sum_{j=1}^J \sum_{b\in B_j} h_b^{(\degexp)}} & \leq 24 e^{-\degexp/4} \sum_{j=1}^J \sum_{b\in B_j} \sum_{\substack{(p,q)\in E_b\\(r,s)\in E_{\bar{b}}}} \int \int \abs{\orb{p}^*(r_1)\orb{q}(r_1)} \frac{1}{\abs{r_1 - r_2}} \abs{\orb{r}^*(r_2)\orb{s}(r_2)} dr_1 dr_2\\
	& =\bigO(e^{-\degexp/4}\nsys^2).
 \end{align*}
 Here, we extended the summation to all 2-electron interactions, and we applied \Cref{lemma:mesh norm} to the atomic-orbital set $\{\abs{\orb{Ii}}\}$.
By choosing $\degexp = \bigO(\log(\nsys/\eps))$ we have a truncation error smaller than $\eps$.

\subsection{Proof of \Cref{lemma:rank1 dirty ancillas}}
\label{sec:proof rank1 dirty ancillas}
By assumption, $\abs{\alpha_p\beta_q}< c=\bigO(1)$, and we can always rescale the coefficients such that $\abs{\alpha_p},\abs{\beta_q}<\sqrt{c}$ without changing the overall operator.   
In a first step, we rewrite $vw = \sum_{p\in[N_b],q\in[N_{\bar{b}}]}\alpha_p\beta_qZ_p Z_q$:
\begin{align*}
	vw = -\sum_{\substack{p\in[N_b]\\q\in[N_{\bar{b}}]}}\alpha_p\beta_q I + \sum_{\substack{p\in[N_b]\\q\in[N_{\bar{b}}]}} \left(\alpha_p \beta_qZ_p + \alpha_p \beta_q Z_q\right) + 4\sum_{\substack{p\in[N_b]\\q\in[N_{\bar{b}}]}} \alpha_p \beta_q \left(\frac{I - Z_p}{2}\right)\left(\frac{I - Z_q}{2}\right).
\end{align*}
The first term just effects a global phase, while the second term is implemented by one layer of single-qubit rotations.
The operators $(I - Z_p)/2$ act on a computational basis state $\ket{x} = \ket{x_0x_1\dots x_{N - 1}}$ of the system as $(I - Z_p)/2\ket{x} = x_p\ket{x}$.
The third term is realized by implementing the transformation $\ket{x}\ket{y} \to \exp(-i4t(x\cdot\alpha)(y\cdot\beta))\ket{x}\ket{y}$ coherently, where $x\cdot\alpha \equiv \sum_{p=0}^{N_b-1}x_p\alpha_p$.
We require two ancillary registers which can be in arbitrary states. 
By linearity of the whole transformation, we may assume them to be in computational basis states $\ket{v}\ket{w}$.
On a high level, the steps are
\begin{align}
	\ket{x}\ket{y} \ket{v}\ket{w}&\longrightarrow \ket{x}\ket{y} \ket{v + x\cdot\alpha} \ket{w + y\cdot\beta} \qquad && \text{(addition to dirty ancillas)} \label{eq:rank1 initialize}\\
	&\longrightarrow e^{-i4t(v + x\cdot\alpha)(w + y\cdot\beta)} \ket{x}\ket{y} \ket{v + x\cdot\alpha} \ket{w + y\cdot\beta} && \text{(phase multiplication)} \label{eq:rank1 multiplication}\\
	& \longrightarrow e^{-i4t(x\cdot\alpha)(y\cdot\beta)} \ket{x} \ket{y} \ket{v}\ket{w}. && \text{(correct \& uncompute)} \label{eq:rank1 uncompute}
\end{align}

It is convenient to rewrite the product in the phase factor into an expression involving integer-valued variables,
\begin{align}\label{eq:phase product integer}
	4t(x\cdot \alpha)(y\cdot \beta) \approx 2\pi \frac{(x\cdot a)(y\cdot b)}{2^{2\lowprecbits}}
\end{align}
with $a = (a_0, a_1,\dots, a_{N_b-1}), b=(b_0, b_1,\dots,b_{N_{\bar{b}}-1})$ being vectors of integers.

To arrive at this expression, we truncate the binary expansions of $\sqrt{2t/\pi}\alpha_p, \sqrt{t2/\pi}\beta_q$ to a common range of bits of precision, and we let $\lowprecbits$ be the number of bits corresponding to precision smaller than 1.
We then define integers $a_p, b_q$ to be the rescaled values of the truncated binary expansions, $a_p = \lfloor2^{\lowprecbits}\sqrt{2t/\pi}\alpha_p \rfloor$ and $b_q = \lfloor2^{\lowprecbits}\sqrt{2t/\pi}\beta_q \rfloor$.
The resulting truncation error in \cref{eq:phase product integer} is of size $\bigO(\nsys^2 t2^{-\lowprecbits})$, and to limit it to $\eps>0$, we set $\lowprecbits = \bigO(\log\frac{\nsys t}{\eps})$.
We can then set the number of ancillas to $\nanc = \bigO(\log\frac{\nsys t}{\eps})$ such that $\nanc \geq 2\lowprecbits$ and $\abs{x\cdot a}, \abs{y\cdot b} < 2^{\nanc}$ are satisfied.

In \cref{circ:addition} a quantum circuit is shown that performs the modular addition $\ket{w}\ket{x}\to\ket{w + x\cdot a \pmod{2^{\nanc}}}\ket{x}$. 
It requires $\nanc$ sequential applications of the multi-controlled phase gate, for which a log-depth circuit construction is presented in \cref{circ:multi phase rot}.
The quantum Fourier transform (QFT) requires $\bigO(\nanc^2)$ total gates and depth $\bigO(\nanc)$.
Therefore the total gate count of the modular addition is $\bigO(\nanc^2 + \nanc\nsys) = \bigO(\log^2\frac{\nsys t}{\eps} + \nsys\log\frac{\nsys t}{\eps})$ and the depth is $\bigO(\nanc\log\nsys) = \bigO(\log\frac{\nsys t}{\eps}\log\nsys)$.

Next, we can compute the phase factor containing the product, $\exp{-i2\pi(v + x\cdot a \pmod{2^{\nanc}})(w + y\cdot b \pmod{2^{\nanc}})/2^{2\lowprecbits}}$, using $\nanc^2$ controlled phase gates in $\nanc$ layers (which is subleading compared to the addition).
Due to the periodicity of the exponential function, the modular arithmetic property $(a \pmod{p})(b \pmod{p})\pmod{p} = ab \pmod{p}$, and the fact that $\nanc \geq 2\lowprecbits$, we can simplify the phase factor to $\exp{-i2\pi(v + x\cdot a)(w + y\cdot b)/2^{2\lowprecbits}}$.
Uncomputation $\ket{w + x\cdot a \pmod{2^{\nanc}}}\ket{x}\to\ket{w}\ket{x}$ is performed by addition of $2^{\nanc} - x\cdot a$.

The phase factor is not correct yet because it still contains the initial values of the dirty ancilla registers, $v$ and $w$.
To correct for that, we simply compute the three phase factors $\exp{i2\pi(v + x\cdot a)w/2^{2\lowprecbits}}, \exp{i2\pi v(w + y\cdot b)/2^{2\lowprecbits}}$ and $\exp{-i2\pi vw/2^{2\lowprecbits}}$.
This is achieved by adding only one of the values $x\cdot a, y\cdot b$ to the ancilla registers (or none) and then applying the controlled phase gates as before.

All in all, this shows how to implement to precision $\eps$
\begin{align*}
	\ket{x}\ket{y} \rightarrow e^{-it(x\cdot \alpha)(y\cdot \beta)}\ket{x}\ket{y}
\end{align*}
with $\bigO(\log\frac{\nsys t}{\eps})$ dirty ancilla qubits in depth $\bigO(\log\frac{\nsys t}{\eps}\log\nsys)$ and with $\bigO(\log^2\frac{\nsys t}{\eps} + \nsys\log\frac{\nsys t}{\eps})$ total gates.

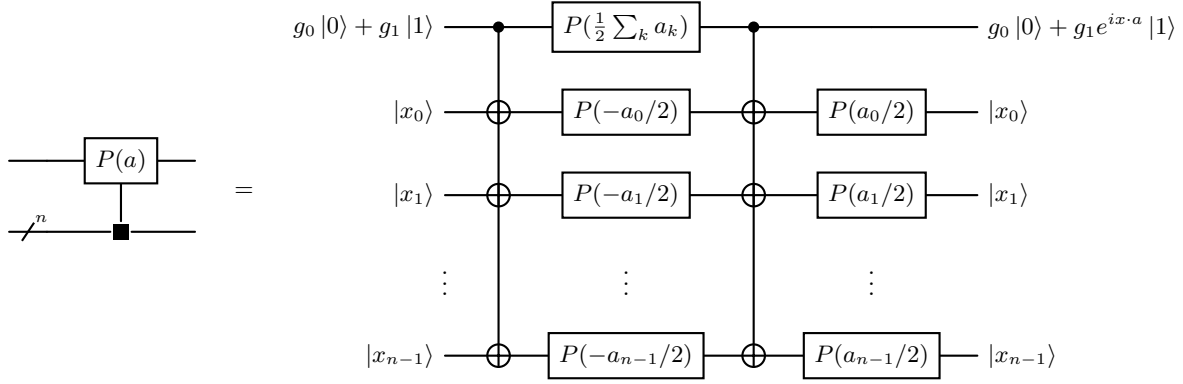
\begin{figure}
\centering
$
\begin{quantikz}
	& & \gate{P(a)} &  \\
	& \qwbundle{n}  & \mctrl{-1} & 
\end{quantikz}
\quad = \quad
\begin{quantikz}
	\lstick{$g_0\ket{0} +  g_1\ket{1}$} & \ctrl{4} & \gate{P(\frac12\sum_k a_k)} & \ctrl{4} & & \rstick{$g_0\ket{0} + g_1 e^{i x\cdot a} \ket{1}$} \\
	\lstick{$\ket{x_0}$} & \targ{} & \gate{P(-a_0/2)} & \targ{} & \gate{P(a_0/2)} & \rstick{$\ket{x_0}$} \\
	\lstick{$\ket{x_1}$} & \targ{} & \gate{P(-a_1/2)} & \targ{} & \gate{P(a_1/2)} & \rstick{$\ket{x_1}$} \\
	\setwiretype{n} \vdots & & \vdots & & \vdots & \\
	\lstick{$\ket{x_{n-1}}$} & \targ{} & \gate{P(-a_{n-1}/2)} & \targ{} & \gate{P(a_{n-1}/2)} & \rstick{$\ket{x_{n-1}}$} 
\end{quantikz}
$
\caption{An implementation of the multi-controlled phase gate in logarithmic depth. 
	The fan-out gate can be implemented in logarithmic depth \cite{ziShallowQuantumCircuit2025}. 
	The single-qubit phase gate is defined as $P(\theta) = \dyad{0} + e^{i\theta}\dyad{1}$.}
\label{circ:multi phase rot}
\end{figure}

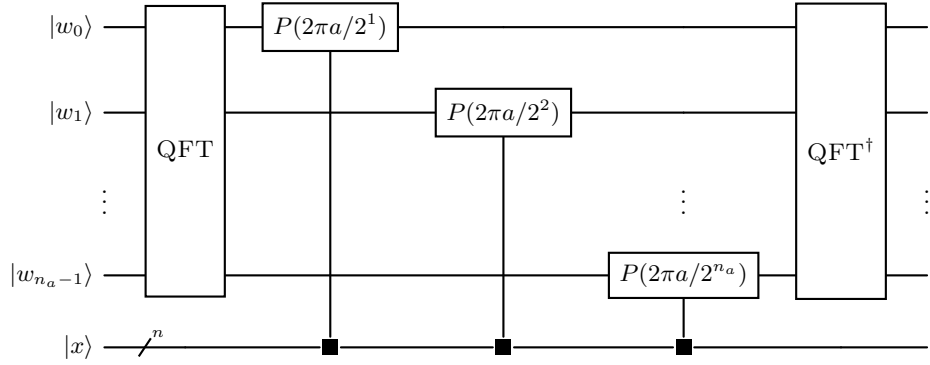
\begin{figure}
\centering
$
\begin{quantikz}
	\lstick{$\ket{w_0}$} & \gate[4]{\qft} & \gate{P(2\pi a/2^1)} & & & \gate[4]{\qft^{\dag}} & \\
	\lstick{$\ket{w_1}$} & & & \gate{P(2\pi a/2^2)} & & &  \\
	\setwiretype{n} \vdots & & & & \vdots & & \vdots \\
	\lstick{$\ket{w_{\nanc-1}}$} & & & & \gate{P(2\pi a/2^{\nanc})} & &  \\
	\lstick{$\ket{x}$} & \qwbundle{n} & \mctrl{-4} & \mctrl{-3} & \mctrl{-1} & &
\end{quantikz}
$
\caption{Quantum circuit performing the addition $\ket{w}\ket{x}\to\ket{w + x\cdot a \pmod{2^{\nanc}}}\ket{x}$ based on Fourier-arithmetic.}
\label{circ:addition}
\end{figure}

\subsection{Proof of \Cref{thm:trotter error}}
For a Hamiltonian decomposed into terms as $H = \sum_{\gamma} H_{\gamma}$ the time-evolution can be approximated by a Trotter product formula which amounts to composing time-evolutions under individual $H_{\gamma}$.
The higher-order Trotter formulas are recursively defined as
\begin{align*}
	S_1(t) = \prod_\gamma e^{-it H_{\gamma}},\qquad S_2(t) = S_1(t/2) S_1^{\dag}(-t/2), \qquad S_{2l}(t) = S_{2l-2}^2(u_l t) S_{2l-2}((1-4u_l) t) S_{2l-2}^2(u_l t)
\end{align*}
with $u_l = 1/(4-4^{1/(2l-1)})$.
The number of stages per Trotter step, \textit{i.e.} the number of iterations through all terms, is $1$ for $S_1(t)$ and $2\times 5^{l-1}$ for $S_{2l}(t)$.
In the main text we showed that a single stage (and therefore a single step) can be implemented in polylogarithmic depth in $\nsys$.
Here we analyze the Trotter error which can be upper bounded in terms of a nested commutator expression \cite{childs_theory_2021}
\begin{align*}
	\norm{e^{-iHt} - S_{k}(t)} &= \bigO\left(\alpha_{\text{comm},k}t^{k+1}\right)\\
	\alpha_{\text{comm},k} &= \sum_{\gamma_1,\gamma_2,\dots,\gamma_{k+1}}\norm{[H_{\gamma_{k+1}},[H_{\gamma_k},\dots[H_{\gamma_2},H_{\gamma_1}]\dots]]},
\end{align*}

We consider $H\in \mathrm{MESH}(\zeta,\beta,\localsize,\delta,A)$; its Trotter error depends on the decomposition into simple terms $H_{\gamma}$, and the one underlying our Trotter step implementation is specified in \Cref{eq:hamil decomp}.
We view the decomposition as a linear combination of elementary operators $T_{pq}$ and $T_{pq}T_{rs}$.
First, we note that in the implementation of $h_2^{\mathrm{long}}$, the actual building blocks in the implementation are the evolutions under products of 1-sparse 1-electron operators, $e^{-iv_b^{\ell m\sigma}w_b^{\ell m \rho}}$.
By construction, all terms in $v_b^{\ell m\sigma}$ and in $w_b^{\ell m \rho}$ commute, therefore
\begin{align*}
	e^{-itv_b^{\ell m\sigma}w_b^{\ell m \rho}} &= \prod_{\substack{(p,q)\in E_b\\(r,s)\in E_{\bar{b}}}} e^{-it v_{pq}^{\ell m\sigma} w_{rs}^{\ell m \rho} T_{pq} T_{rs}}
\end{align*}
and hence our Trotter decomposition does in fact reduce to time evolutions under the aforementioned elementary operators. 
We note that in the long-range decomposition
\begin{align}
	h_b = \frac12\sum_{\substack{(p,q)\in E_b\\(r,s)\in E_{\bar{b}}}} h_{pqrs} T_{pq} T_{rs} \approx \frac12\sum_{\ell=0}^{\degexp}\sum_{m=-\ell}^{\ell} \sum_{\substack{(p,q)\in E_b\\(r,s)\in E_{\bar{b}}}}\sum_{\sigma,\rho=1}^{\Delta+2} v_{pq}^{\ell m\sigma} w_{rs}^{\ell m\rho} T_{pq} T_{rs}
\end{align}
an operator $T_{pq} T_{rs}$ appears multiple times, once for each pair of values $\ell,m$.
Its cumulative weight is at most
\begin{align*}
	\abs{h_{pqrs}'} :=& \sum_{\ell=0}^{\degexp}\sum_{m=-\ell}^{\ell}\sum_{\sigma,\rho=1}^{\Delta+2} \abs{v_{pq}^{\ell m\sigma} w_{rs}^{\ell m\rho}}
	 = \sum_{\ell=0}^{\degexp}\sum_{m=-\ell}^{\ell} \abs{v_{pq}^{\ell m} w_{rs}^{\ell m}}\\
	=& \sum_{\ell=0}^{\degexp} \sum_{m=-\ell}^{\ell}\abs{\int \int \orb{p}^*(r_1)\orb{q}(r_1) \abs{r_1}^{\ell} Y_{\ell m}(\theta_1,\varphi_1)\frac{Y_{\ell m}(\theta_2,\varphi_2)}{\abs{r_2}^{\ell+1}}\orb{r}^*(r_2)\orb{s}(r_2) dr_1 dr_2}\\
	\leq& \int \int \abs{\orb{p}^*(r_1)} \abs{\orb{q}(r_1)} \frac{1}{\abs{r_2}} \sum_{\ell=0}^{\degexp} \left(\frac{\abs{r_1}}{\abs{r_2}}\right)^{\ell} \abs{\orb{r}^*(r_2)} \abs{\orb{s}(r_2)} dr_1 dr_2\\
	\leq& 24 \int \int \abs{\orb{p}^*(r_1)} \abs{\orb{q}(r_1)}\frac{1}{\abs{r_1 - r_2}} \abs{\orb{r}^*(r_2)} \abs{\orb{s}(r_2)} dr_1 dr_2
\end{align*}
where we used Unsöld's theorem, $\abs{r_1}/\abs{r_2}<8/11$, $\abs{r_1 - r_2}<7d_j + 2\rupper < 8d_j$ and $\abs{r_2}>11d_j/8$.
Because the commutator is linear in its arguments, we can equivalently study $\alpha_{\mathrm{comm},k}$ for the decomposition of $H$ into terms $H_{\gamma}\in \{h_{pq}T_{pq}\}_{pq}\cup\{\abs*{h'_{pqrs}}T_{pq}T_{rs}\}_{pqrs}$.
Concretely, each term $T_{pq}T_{rs}$ appears only once, but with increased weight $\abs{h_{pqrs}'}$.
We abbreviate the nested commutator expression by
\begin{align}\label{eq:W^k}
	W^{(k)} = \sum_{\gamma_0,\gamma_1,\dots,\gamma_{k}}[H_{\gamma_{k}},[H_{\gamma_{k-1}},\dots[H_{\gamma_1},H_{\gamma_0}]]],
\end{align}
and we introduce the skew-Hermitian operators $S_{pq} = \cre{p}\des{q} - \cre{q}\des{p}$ which are related to $T_{pq} = \cre{p}\des{q} + \cre{q}\des{p}$ via the commutation relationships
\begin{align}\label{eq:commutation relations T S}
\begin{split}
	[T_{pq},T_{rs}] &= \delta_{qr} S_{ps} + \delta_{pr} S_{qs} + \delta_{qs} S_{pr} + \delta_{ps} S_{qr}\\
	[S_{pq},T_{rs}] &= \delta_{qr}T_{ps} - \delta_{pr} T_{sq} + \delta_{qs} T_{pr} - \delta_{ps} T_{qr}.
\end{split}
\end{align}
\begin{lemma}\label[lemma]{commutator}
	The nested commutator \cref{eq:W^k} is of the form
	\begin{align*}
	W^{(k)} = \sum_{\vec{p} \vec{q}} w_{\vec{p} \vec{q}}^{(k)} W_{p_0q_0}W_{p_1q_1}\dots W_{p_{k+1}q_{k+1}}
	\end{align*}
	where $w_{\vec{p} \vec{q}}^{(k)}$ is a coefficient indexed by vectors of orbital indices $\vec{p},\vec{q}$, and $W_{p_iq_i} \in \{T_{pq}\}_{pq} \cup \{S_{pq}\}_{pq}$.
	Moreover, the 1-norm of the coefficients is bounded by
	\begin{align*}
		\norm*{w^{(k)}}_1 = 
		\begin{cases}
			\bigO(8^k(k+1)!\nsys \log^{k+1}\nsys)& d=1 \\
			\bigO(8^k(k+1)!\nsys^{(k+1)(\frac{d-1}{d}) + 1})&d>1
		\end{cases}.
	\end{align*}
\end{lemma}
\begin{proof}
	The proof is by induction over $k$.
	For $k=0$ we have $W^{(0)} = H$ which by definition is of the claimed form.
	We can apply \Cref{lemma:mesh norm} to the atomic orbitals $\{\abs{\orb{p}}\}_p$ and multiply by 24 to bound
	\begin{align*}
		\norm*{w^{(0)}}_1 = \sum_{pq}\abs{h_{pq}} + \sum_{pqrs}\abs{h'_{pqrs}} = \begin{cases} \bigO(\nsys\log\nsys)&d=1\\ \bigO(\nsys^{\frac{2d-1}{d}})&d>1 \end{cases}.
	\end{align*}
	Now let $k>0$ and assume the claim holds for $k-1$.
	We will treat the 1- and 2-electron terms entering $W^{(k)}$ separately,
	\begin{align*}
		W^{(k)} = \sum_{pq} h_{pq} [T_{pq},W^{(k-1)}] + \sum_{pqrs} \abs{h'_{pqrs}}[T_{pq}T_{rs},W^{(k-1)}].
	\end{align*}
	From the induction hypothesis and the commutation relations \cref{eq:commutation relations T S}, it follows that 
	\begin{align*}
		\sum_{pq}[h_{pq}T_{pq},W^{(k-1)}] &= \sum_{pq}h_{pq}\sum_{\vec{p} \vec{q}} w_{\vec{p} \vec{q}}^{(k-1)} \sum_{i=0}^{k} W_{p_0q_0}\cdots W_{p_{i-1}q_{i-1}}[T_{pq},W_{p_iq_i}]W_{p_{i+1}q_{i+1}}\cdots W_{p_kq_k}
	\end{align*}
	is a sum over products of at most $k+1$ operators in $\{T_{pq}\}_{pq}\cup\{S_{pq}\}_{pq}$.
	Then, by \cref{lemma:mesh norm}, the corresponding coefficients are upper bounded by
	\begin{align*}
		\sum_{i=0}^{k} \sum_{pq} \sum_{\vec{p} \vec{q}} \abs{h_{pq}} \abs*{w_{\vec{p} \vec{q}}^{(k-1)}} &(\delta_{pp_i} + \delta_{pq_i} + \delta_{qp_i} + \delta_{qq_i})
		 \leq 4(k+1) \norm*{w^{(k-1)}}_1 \max_p \sum_q \abs*{h_{pq}}\\
		 &= \begin{cases}
		 	\bigO(4^k(k+1)!\nsys\log^{k+1}\nsys)&d=1\\
			\bigO(4^k(k+1)!\nsys^{(k+1)(\frac{d-1}{d}) + 1)} &d>1
		  \end{cases}.
	\end{align*}
	Similarly, $[T_{pq}T_{rs},W^{(k-1)}] = [T_{pq},W^{(k-1)}]T_{rs} + T_{pq}[T_{rs},W^{(k-1)}]$ is a sum over products of at most $k+2$ operators in $\{T_{pq}\}_{pq}\cup\{S_{pq}\}_{pq}$.
	Using \cref{lemma:mesh norm}, the coefficients are upper bounded by
	\begin{align*}
		2\sum_{i=0}^k\sum_{pqrs} \sum_{\vec{p} \vec{q}} \abs{h'_{pqrs}} \abs*{w_{\vec{p},\vec{q}}^{(k-1)}} &(\delta_{pp_i} + \delta_{pq_i} + \delta_{qp_i} + \delta_{qq_i}) \leq 8(k+1) \norm*{w^{(k-1)}}_1 \max_p \sum_{qrs} \abs*{h'_{pqrs}}\\
		 &= \begin{cases}
		 	\bigO(8^k(k+1)!\nsys\log^{k+1}\nsys)\quad & \text{if} \quad d=1\\
			\bigO(8^k(k+1)!\nsys^{(k+1)(\frac{d-1}{d}) + 1)} \quad & \text{if} \quad d>1
		  \end{cases}.
	\end{align*}
	\end{proof}
Each term $(\gamma_0,\gamma_1,\cdots,\gamma_k)$ in $W^{(k)}$ corresponds to a term  $(\vec{p}\vec{q})$ in \Cref{commutator}. 
Because of $\norm{T_{pq}} = \norm{S_{pq}} = 1$, we then obtain
\begin{align*}
	\alpha_{\mathrm{comm},k} = 
	\begin{cases}
		\bigO(8^k(k+1)!\nsys\log^{k+1}\nsys)\quad & \text{if} \quad d=1\\
		\bigO(8^k(k+1)!\nsys^{(k+1)(\frac{d-1}{d}) + 1)} \quad & \text{if} \quad d>1
	\end{cases}
\end{align*}
which concludes the proof of the Theorem.


\end{document}